\documentclass[a4paper,11pt]{article}
\input{Preamble.tex}

\begin{document}
\hypersetup{pageanchor=false}
\title{Induced-Minor-Free Graphs: Separator Theorem, Subexponential Algorithms, and Improved Hardness of Recognition\thanks{The research leading to these results has received funding from the Research Council of Norway via the project BWCA (grant no. 314528) and NSF award CCF-2008838.}
}

\author{Tuukka Korhonen\thanks{Department of Informatics, University of Bergen, Norway. \texttt{tuukka.korhonen@uib.no}}
\and
Daniel Lokshtanov\thanks{Department of Computer Science, University of California Santa Barbara, USA. \texttt{daniello@ucsb.edu}}
}

\maketitle

\thispagestyle{empty}

\begin{abstract}
A graph $G$ contains a graph $H$ as an induced minor if $H$ can be obtained from $G$ by vertex deletions and edge contractions.
The class of $H$-induced-minor-free graphs generalizes the class of $H$-minor-free graphs, but unlike $H$-minor-free graphs, it can contain dense graphs.
We show that if an $n$-vertex $m$-edge graph $G$ does not contain a graph $H$ as an induced minor, then it has a balanced vertex separator of size $\OO_{H}(\sqrt{m})$, where the $\OO_{H}(\cdot)$-notation hides factors depending on $H$.
More precisely, our upper bound for the size of the balanced separator is $\OO(\min(|V(H)|^2, \log n) \cdot \sqrt{|V(H)|+|E(H)|} \cdot \sqrt{m})$.
We give an algorithm for finding either an induced minor model of $H$ in $G$ or such a separator in randomized polynomial-time.
We apply this to obtain subexponential $2^{\OO_{H}(n^{2/3} \log n)}$ time algorithms on $H$-induced-minor-free graphs for a large class of problems including maximum independent set, minimum feedback vertex set, 3-coloring, and planarization.

For graphs $H$ where every edge is incident to a vertex of degree at most 2, our results imply a $2^{\OO_{H}(n^{2/3} \log n)}$ time algorithm for testing if $G$ contains $H$ as an induced minor.
Our second main result is that there exists a fixed tree $T$, so that there is no $2^{o(n/\log^3 n)}$ time algorithm for testing if a given $n$-vertex graph contains $T$ as an induced minor unless the Exponential Time Hypothesis (ETH) fails.
Our reduction also gives NP-hardness, which solves an open problem asked by Fellows, Kratochv{\'{\i}}l, Middendorf, and Pfeiffer~[Algorithmica,~1995], who asked if there exists a fixed planar graph $H$ so that testing for $H$ as an induced minor is NP-hard.
\end{abstract}    

\newpage
\hypersetup{pageanchor=true}
\pagestyle{plain}
\pagenumbering{arabic}

\section{Introduction}
Graph classes that exclude some fixed graph as a minor have been widely studied in algorithmic graph theory.
Graphs that exclude a planar graph as a minor correspond to graphs of bounded treewidth~\cite{DBLP:journals/jct/RobertsonS86}, which is among the most popular graph parameters in algorithms.
Graphs excluding a fixed non-planar graph as a minor have turned out to be similar to planar graphs in their algorithmic properties, with a large number of algorithms generalized from planar graphs to them~\cite{DBLP:conf/stoc/AlonST90, DemaineHK11,DemaineFHT05jacm}.

In this paper we study graph classes that exclude some fixed graph as an \emph{induced minor}.
A graph $G$ contains $H$ as an induced minor if $H$ can be obtained from $G$ by vertex deletions and edge contractions.
Here we assume that all graphs are simple, in particular no self-loops or parallel edges are created by contractions.
Note that if $G$ contains $H$ as an induced minor, then $G$ contains $H$ also as a minor, and therefore the class of $H$-induced-minor-free graphs is a superclass of $H$-minor-free graphs for every $H$.

Graphs excluding a graph $H$ as a minor are necessarily sparse in the sense that their average degree is at most $\OO(|V(H)| \cdot \sqrt{\log |V(H)|})$~\cite{DBLP:journals/combinatorica/Kostochka84,thomason_1984}.
However, by excluding a graph $H$ as an induced minor instead of a minor, one can obtain well-structured dense graphs classes.
For example, $C_4$-induced-minor-free graphs are exactly the chordal graphs, which despite containing all complete graphs admit strong algorithmic properties~\cite{DBLP:journals/siamcomp/Gavril72}.

While induced minors have been occasionally studied since the 80's~\cite{DBLP:journals/algorithmica/FellowsKMP95,DBLP:journals/jda/FialaKP12,DBLP:journals/jct/Kratochvil91,matouvsek1988polynomial,DBLP:journals/jct/Thomas85}, recently there has been a growing interest in understanding the structural and algorithmic properties of graphs that exclude some particular graph $H$ or a family of graphs $\mathcal{H}$ as induced minors~\cite{abrishami2021induced3,DBLP:conf/soda/BonamyBDEGHTW23,DBLP:journals/corr/abs-2302-08182,DBLP:conf/focs/GartlandL20,DBLP:conf/stoc/GartlandLPPR21,DBLP:journals/combinatorics/Hickingbotham23,DBLP:journals/jctb/Korhonen23}\footnote{Note that often the condition of excluding a family of graphs as induced subgraphs can be phrased as excluding a single graph as an induced minor. For example, excluding all cycles of length $>t$ as induced subgraphs corresponds to excluding $C_{t+1}$ as an induced minor.}.
Most notably, polynomial-time and quasipolynomial-time algorithms for independent set on some of such classes have been obtained~\cite{DBLP:conf/soda/BonamyBDEGHTW23,DBLP:journals/corr/abs-2302-08182,DBLP:conf/focs/GartlandL20,DBLP:conf/stoc/GartlandLPPR21}, and some of such classes with additional sparsity conditions have been shown to have bounded or logarithmic treewidth~\cite{abrishami2021induced3,DBLP:conf/soda/BonamyBDEGHTW23,DBLP:journals/jctb/Korhonen23} or pathwidth~\cite{DBLP:journals/combinatorics/Hickingbotham23}.

All of the aforementioned algorithmic results concern classes that exclude some fixed planar graph as an induced minor.
In this paper, we study the class of $H$-induced-minor-free graphs for an arbitrary non-planar graph $H$.
As these classes contain all planar graphs, we cannot hope to obtain (quasi)polynomial-time algorithms for most of the NP-hard graphs problems on them.
However, planar and $H$-minor-free graphs admit subexponential $2^{\OO(\sqrt{n})}$ time algorithms for many NP-hard graph problems thanks to separator theorems that state that they have balanced vertex separators of size $\OO(\sqrt{n})$~\cite{DBLP:conf/stoc/AlonST90,doi:10.1137/0136016,DBLP:journals/siamcomp/LiptonT80}.
Our first main result is a separator theorem for $H$-induced-minor-free graphs.

\begin{restatable}{theorem}{septheorem}
\label{thm:septheorem}
There is a randomized polynomial-time algorithm that is given graphs $G$ and $H$, and outputs either an induced minor model of $H$ in $G$, or a balanced separator of $G$ of size at most $\OO(\min(\log |V(G)|, |V(H)|^2) \cdot \sqrt{|V(H)| + |E(H)|} \cdot \sqrt{|E(G)|})$.
\end{restatable}

In particular, $H$-induced-minor-free graphs with $m$ edges have balanced separators of size at most $\OO_H(\sqrt{m})$, where the $\OO_{H}(\cdot)$-notation hides factors depending on $H$.
The dependence $\sqrt{m}$ on the number of edges $m$ is clearly tight, and in fact a ``thickened grid''-construction shows that for \emph{every} value of maximum degree $\Delta \ge 3$, a separator of size $\OO_H(\sqrt{\Delta \cdot n})$ is the best possible.
We also obtain a nearly-optimal dependence on $H$ in the regime where $H$ is large compared to $G$.
In particular, \Cref{thm:septheorem} implies that there is a constant $c > 0$ so that any bounded-degree expander graph with $n$ vertices contains all graphs with at most $c \cdot n / \log^2 n$ edges and vertices as induced minors.
This can be compared to a similar results about minors~\cite{DBLP:journals/corr/abs-1901-09349,DBLP:conf/bcc/Krivelevich19}.

While \Cref{thm:septheorem} does not imply subexponential algorithms as directly as the separator theorem for $H$-minor-free graphs, there is a well-known method for exploiting such separator theorems for subexponential time algorithms for problems such as \textsc{Maximum Independent Set}~\cite{DBLP:conf/soda/FoxP11,DBLP:journals/algorithmica/BonnetR19}: As long as the graph contains a vertex of degree more than $n^{1/3}$, branch on this vertex, and after the maximum degree is at most $n^{1/3}$, use the separator theorem to argue that treewidth is at most $\OO_H(n^{2/3})$ and solve the problem by dynamic programming on treewidth.
Generalizing this technique to other problems yields the following algorithmic applications.

\begin{corollary}
\label{cor:algos}
The following problems can be solved in $2^{\OO_H(n^{2/3} \log n)}$ time on $H$-induced-minor-free graphs:
\begin{enumerate}\setlength\itemsep{-.7pt}
\item\label{cor:algos:mis} \textsc{Maximum Independent Set}
\item\label{cor:algos:mim} \textsc{Maximum Induced Matching}
\item\label{cor:algos:fvs} \textsc{Minimum Feedback Vertex Set}
\item\label{cor:algos:plan} \textsc{Planarization}
\item\label{cor:algos:fmd} $\mathcal{F}$-\textsc{Minor-Deletion} for any fixed family $\mathcal{F}$ of connected graphs
\item\label{cor:algos:3col} $3$-\textsc{Coloring}
\end{enumerate}
\end{corollary}

We obtain the algorithms of \Cref{cor:algos:mis,cor:algos:mim,cor:algos:fvs,cor:algos:plan,cor:algos:fmd} of \Cref{cor:algos} by giving a meta-theorem that states that $\OO(\sqrt{m})$ separator theorems and $2^{\OO(\tw \log \tw)} n^{\OO(1)}$ time algorithms (where $\tw$ is the treewidth) can be combined to yield $2^{\OO(n^{2/3} \log n)}$ time algorithms for problems where the task is to find an induced subgraph of bounded degeneracy that satisfies some property.
This improves over a similar meta-theorem of~\cite{DBLP:journals/algorithmica/NovotnaOPRLW21} that would yield $2^{\OO(n^{3/4} \polylog\ n)}$ time algorithms, although for a slightly more general set of problems.
The approach for $3$-\textsc{Coloring} is from~\cite{DBLP:journals/algorithmica/BonnetR19}.

To the best of our knowledge, no subexponential-time algorithms for $H$-induced-minor-free graphs for arbitrary $H$ were known prior to our resuls.
However, some classes of graphs that can be shown to exclude some fixed graph as an induced minor have been considered, originating from the geometric setting~\cite{DBLP:conf/soda/FoxP11,DBLP:conf/soda/LokshtanovPSXZ22}.
Most notably, the structural and algorithmic properties of \emph{string graphs}, which are the intersection graphs of arbitrary regions in the plane, have been actively studied~\cite{FOX20081070,DBLP:journals/cpc/FoxP10,DBLP:conf/soda/FoxP11,DBLP:journals/cpc/FoxP14,DBLP:journals/jct/Kratochvil91,DBLP:journals/jct/Kratochvil91a,DBLP:conf/innovations/Lee17,DBLP:journals/cpc/Matousek14,sinden1966topology,DBLP:conf/stoc/SchaeferS01,DBLP:conf/stoc/SchaeferSS02}.
String graphs contain all planar graphs~\cite{sinden1966topology}, but exclude the subdivided $K_5$ as an induced minor.
Lee~\cite{DBLP:conf/innovations/Lee17} gave a $\OO(\sqrt{m})$ separator theorem for string graphs, improving over a $\OO(m^{3/4} \sqrt{\log m})$ separator theorem of Fox and Pach~\cite{DBLP:journals/cpc/FoxP10} and a $\OO(\sqrt{m} \log m)$ separator theorem of Matousek~\cite{DBLP:journals/cpc/Matousek14}.
Lee in fact showed a more general result, showing that such a $\OO(\sqrt{m})$ separator theorem holds also for the intersection graphs of connected subgraphs on $H$-minor-free graphs, which generalize the string graphs.
As such graphs exclude a subdivided $H$ as an induced minor, our \Cref{thm:septheorem} further generalizes the result of Lee.

The applications of the separator theorem of Lee~\cite{DBLP:conf/innovations/Lee17} to subexponential time algorithms on string graphs were studied by Bonnet and Rzazewski~\cite{DBLP:journals/algorithmica/BonnetR19}.
They provided $2^{\OO(n^{2/3} \polylog\ n)}$ upper bounds for $3$-\textsc{Coloring}, \textsc{Minimum Feedback Vertex Set}, and \textsc{Maximum Induced Matching} that were generalized by our \Cref{cor:algos}.
They also obtained lower bounds that under the Exponential Time Hypothesis (ETH) exclude $2^{o(n)}$ time algorithms on string graphs for $k$-\textsc{Coloring} for every $k \ge 4$, \textsc{Minimum Dominating Set}, and \textsc{Maximum Clique}.
These lower bounds of course also apply for $H$-induced-minor-free graphs.
For problems admitting subexponential time algorithms on string graphs, no better than $2^{\OO(n^{2/3})}$ time algorithms nor better than $2^{o(\sqrt{n})}$ lower bounds are known (the lower bounds follow from lower bounds on planar graphs~\cite{DBLP:journals/siamcomp/Lichtenstein82}).
Marx and Pilipczuk~\cite{DBLP:conf/esa/MarxP15} gave a $2^{\OO(\sqrt{n} \log n)} p^{\OO(1)}$ time algorithm for \textsc{Maximum Independent Set} on string graphs that are given with a representation that has $p$ vertices.
However, there exists string graphs whose representation requires $2^{\Omega(n)}$ vertices~\cite{DBLP:journals/jct/KratochvilM91}.

With techniques used for proving \Cref{cor:algos}, we also obtain the following result about testing whether a given graph $G$ contains a given graph $H$ as an induced minor.

\begin{restatable}{corollary}{testingcor}
\label{cor:indminregoc}
Let $H$ be a graph where every edge is incident to a vertex of degree at most $2$.
There is a $2^{\OO_{H}(n^{2/3} \log n)}$ time algorithm for testing if a given $n$-vertex graph contains $H$ as an induced minor.
\end{restatable}

In particular, \Cref{cor:indminregoc} follows from the fact that minimal induced minor models of such graphs $H$  have bounded degeneracy.
We remark that \Cref{cor:indminregoc} allows us to ``approximately'' test whether a given graph is $H$-induced-minor-free for arbitrary $H$ in the following sense:
The subdivided clique is a universal induced minor in that if a graph $G$ excludes $H$ as an induced minor, then it must also exclude the subdivided $K_{|V(H)|}$ as an induced minor.
Therefore, as every edge of the subdivided $K_{|V(H)|}$ is adjacent to a degree-2 vertex, we can for any fixed graph $H$ in time $2^{\OO_{H}(n^{2/3} \log n)}$ report either that $G$ contains $H$ as an induced minor, or that $G$ excludes the subdivided $K_{|V(H)|}$ as an induced minor.

Our second main result concerns the question of whether an algorithm similar to \Cref{cor:indminregoc} could be given for all graphs $H$.
We show that the answer is negative, even when $H$ is a tree.

\begin{restatable}{theorem}{hardnesstheorem}
\label{thm:hardnesstheorem}
There exists a fixed tree $T$, so that assuming ETH, there is no $2^{o(n/\log^3 n)}$ time algorithm for testing if a given $n$-vertex graph contains $T$ as an induced minor.
\end{restatable}

As expected, the proof of \Cref{thm:hardnesstheorem} also implies NP-completeness of testing whether a given graph contains the fixed tree $T$ as an induced minor.
This answers two open problems of Fellows, Kratochv{\'{\i}}l, Middendorf, and Pfeiffer~\cite{DBLP:journals/algorithmica/FellowsKMP95}, who asked (1) ``Is there a planar graph $H$ for which $H$-induced minor testing is NP-complete?'' and (2) ``Can $H$-induced minor testing always be done in polynomial time when $H$ is a tree?''.
In their paper, Fellows, Kratochv{\'{\i}}l, Middendorf, and Pfeiffer~\cite{DBLP:journals/algorithmica/FellowsKMP95} gave NP-completeness of $H$-induced minor testing for some fixed $H$, but their proof was strongly based on planarity in a way that would not work for planar $H$ and could not yield better than $2^{o(\sqrt{n})}$ lower bounds.
Fiala, Kaminski, and Paulusma~\cite{DBLP:journals/jda/FialaKP12} gave a polynomial-time algorithm for $H$-induced minor testing for all 7-vertex forests $H$ except for the tree obtained by gluing together two three-leaf stars from leaves, for which the complexity still remains open~\cite[Problem 4.3]{chudnovsky_et_al:DagRep.12.11.109}.

\paragraph{Organization.}
In \Cref{sec:overview} we give an overview of the proofs of our results.
In \Cref{sec:prelims} we present notation and preliminary results.
In \Cref{sec:septheorem} we prove our separator theorem, i.e., \Cref{thm:septheorem}.
In \Cref{sec:algos} we give subexponential algorithms by using \Cref{sec:septheorem}, in particular, we prove \Cref{cor:algos,cor:indminregoc}.
In \Cref{sec:hardness} we give the hardness result for testing if a graph contains a fixed tree as an induced minor, i.e., \Cref{thm:hardnesstheorem}.
Finally, we conclude in \Cref{sec:concl} with some additional remarks and open questions.

\section{Overview}
\label{sec:overview}

\subsection{Overview of the Separator Theorem}\label{sec:sepOver}
We aim to prove \Cref{thm:septheorem}, namely that there is a randomized polynomial-time algorithm that is given graphs $G$ and $H$, and outputs either an induced minor model of $H$ in $G$, or a balanced separator of $G$ of size at most $\OO(\min(\log |V(G)|, |V(H)|^2) \cdot \sqrt{|V(H)| + |E(H)|} \cdot \sqrt{|E(G)|})$.
%
%
Here a balanced separator is a vertex set $S$ so that every component of $G \setminus S$ has at most $2n/3$ vertices. For simplicity we assume that $H$ has no isolated vertices, and that therefore $|V(H)| \leq 2|E(H)|$. 
%
%

We first use well-established techniques~\cite{FeigeHL08,DBLP:conf/innovations/Lee17,DBLP:journals/jacm/LeightonR99} to either obtain a separator of the desired size, or an induced minor model of $H$ in $G$, or an induced subgraph $G'$ of $G$ on at least $2|V(G)|/3$ vertices and a {\em concurrent flow} in $G'$ with congestion at most $\gamma = \frac{1}{c \cdot \sqrt{|E(G)||E(H)|}}$ where we can pick the constant $c$ to be as large as we want, but independent of $G$ and $H$. For either one of the first two outcomes we are already done, so we may assume that we are in the last case - namely that we have a concurrent flow. What is a concurrent flow in a graph $G$? 

One of the equivalent definitions is that it is a probability distribution on the set ${\cal P}$ of all paths in the graph $G$, such that if we sample a path $P$ from ${\cal P}$ according to this distribution then the (ordered) pair of endpoints of $P$ is uniformly distributed among all vertex pairs in $G$. The {\em congestion} $\gamma$ 
is then the maximum (taken over all vertices $v$ of $G$) of the probability that $v$ is on the sampled path $P$. We remark that in the actual proof we will use a slightly different definition of concurrent flows and congestion where all numbers are scaled by a factor $|V(G)|^2$.

We prove that if a graph $G$ contains a concurrent flow with congestion $ \frac{1}{c \cdot \sqrt{|E(G)||E(H)|}}$ for a sufficiently large constant $c$, then it contains an induced minor model of $H$ in $G$. We have a concurrent flow in an induced subgraph $G'$ of $G$ rather than $G$ itself. However $G'$ contains at least $2/3$ of the vertices of $G$ and so, for 
the sake of obtaining an induced minor model of $H$ we may assume that $G'$ is all of $G$ (the difference is only in the constant factors). 
%
We can also assume without loss of generality that $H$ has maximum degree $3$, as every $H$ is an induced minor of a graph $H'$ with maximum degree $3$ and with $|E(H')| \leq 3|E(H)|$.

To find an induced minor model of $H$, we will instead find a {\em near-induced minor model} of $\dotdot{H}$. Here $\dotdot{H}$ is the graph obtained from $H$ by subdividing every edge twice, and the concept of near-induced minor models is a key novel element in our proof. 
A near-induced minor model of $\dotdot{H}$ in $G$ is a mapping $\phi$ that assigns to each vertex $x$ of $\dotdot{H}$ a vertex $\phi(x)$ in $G$, and to every edge $xy \in E(\dotdot{H})$ the vertex set $\phi(xy)$ of a path from $\phi(x)$ to $\phi(y)$ in $G$. The mapping additionally needs to satisfy the following constraint: for every pair $xy$ and $x'y'$ of edges in  $\dotdot{H}$  that do not share any endpoints, the paths $\phi(xy)$ and $\phi(x'y')$ are vertex-disjoint, and there is no edge in $G$ from $\phi(xy)$ and $\phi(x'y')$. We make no constraints on how paths corresponding to edges of $\dotdot{H}$ that do share an endpoint interact with one another. Thus, for edges $xy$ and $xy'$ in $E(\dotdot{H})$ the paths $\phi(xy)$ and $\phi(xy')$ might intersect. 
The reason for the name {\em near-induced minor model} is that a near-induced minor model of $\dotdot{H}$ in $G$ does not necessarily imply that $G$ contains $\dotdot{H}$ as an induced minor. On the other hand it is not too difficult to show that if $G$ contains a near-induced minor model of $\dotdot{H}$, then $G$ contains an induced minor model of $H$. We can now focus our attention on exhibiting a near-induced minor model of $\dotdot{H}$ in $G$.

We aim to exhibit a near-induced minor model $\phi$ of $\dotdot{H}$ in $G$. We sample for each vertex $x \in V(\dotdot{H})$ a vertex $\phi(x) \in V(G)$ uniformly at random, and then for every edge $xy \in E(\dotdot{H})$ we sample a path $\phi(xy)$ from $\phi(x)$ to $\phi(y)$ using the concurrent flow distribution conditioned on the endpoints of the path being $\phi(x)$ and $\phi(y)$. For $\phi$ to be a a near-induced minor model of $\dotdot{H}$ in $G$ we need that for every pair $xy$ and $x'y'$ of edges in  $\dotdot{H}$  that do not share any endpoints, the paths $\phi(xy)$ and $\phi(x'y')$ are vertex-disjoint, and there is no edge in $G$ from $\phi(xy)$ and $\phi(x'y')$.

A few remarks are in order. First  for every edge $xy \in E(\dotdot{H})$ the path $\phi(xy)$ is distributed precisely as a path drawn from ${\cal P}$ according to the concurrent flow. This follows directly from the definition of concurrent flows: that if we sample a path $P$ from ${\cal P}$ according to this distribution then the (ordered) pair of endpoints of $P$ is uniformly distributed among all vertex pairs in $G$. Second, for every pair of edges $xy, x'y' \in E(\dotdot{H})$  that do not share any endpoints the paths $\phi(xy)$ and $\phi(x'y')$ are independent random variables. 

For each (unordered) pair of edges $xy, x'y' \in E(\dotdot{H})$ which do not share an endpoint we define a bad event: that the paths $\phi(xy)$ and $\phi(x'y')$ {\em collide}, namely that they are not vertex disjoint or that there is an edge in $G$ with one endpoint in $\phi(xy)$ and the other in $\phi(x'y')$. The mapping $\phi$ is a near-induced minor model of $\dotdot{H}$ in $G$ if and only if none of the bad events occur. Thus, to prove the existence of a near-induced minor model of $\dotdot{H}$ in $G$ it suffices to show that with non-zero probability none of the bad events occur. Towards this end we will use the Lov\'{a}sz Local Lemma~\cite{erdos1975problems} (or rather its algorithmic counterpart, as proved by Moser and Tardos~\cite{DBLP:journals/jacm/MoserT10}).
The Lov\'{a}sz Local Lemma states that if you have some bad events, and each bad event occurs with probability at most $p$, and depends on at most $d$ other bad events, and $4dp \leq 1$, then the probability that none of the bad events occur is non-zero. To apply the Lov\'{a}sz Local Lemma to our set of bad events we need to upper bound $p$ and $d$.

First we bound the dependency degree $d$. Consider two bad events, one corresponding to the collision of the edges $x_1 y_1$ and $x_1'y_1'$, and the other corresponding to the collision of  $x_2y_2$ and $x_2'y_2'$.
The two events are independent unless $\{x_1, y_1, x_1', y_1'\} \cap \{x_2, y_2, x_2', y_2'\}$ is non-empty. Since $\dotdot{H}$ has maximum degree $3$ this bounds the dependency degree $d$ by $10|E(\dotdot{H})|$.

Now we upper bound the probability $p$ of each individual bad event occurring. Consider the bad event corresponding to the collision of the edges $xy$ and $x'y'$. For each $uv \in E(G)$ the probability that $u$ is in $\phi(xy)$ is at most the congestion $\frac{1}{c \cdot \sqrt{|E(G)||E(H)|}}$, and the same upper bound holds for the probability that $v$ is $\phi(x'y')$. Since $xy$ and $x'y'$ do not share any endpoints it follows that the probability that  $u$ is in $\phi(xy)$ and $v$ is $\phi(x'y')$ is at most $\frac{1}{c^2|E(G)||E(H)|}$. An identical argument shows that for every vertex $v \in V(G)$ the probability that  $v$ is in $\phi(xy)  \cap \phi(x'y')$ is upper bounded by $\frac{1}{c^2|E(G)||E(H)|}$. Taking a union bound over all vertices and edges of $G$ shows that the probability that  $xy$ and $x'y'$ collide is at most $\frac{|V(G)|+|E(G)|}{c^2|E(G)||E(H)|} \leq \frac{3}{c^2|E(H)|}$. Here we use the assumption that $G$ has no isolated vertices to bound $|V(G)|+|E(G)|$ by $3|E(G)|$.

Since $|E(H)|$ and $|E(\dotdot{H})|$ differ by a factor of $3$ it follows that setting $c = 20$ implies that
$$4dp \leq 4 \cdot 10|E(\dotdot{H})| \cdot \frac{3}{c^2|E(H)|} < 1\mbox{.}$$
Hence the  Lov\'{a}sz Local Lemma implies that with non-zero probability none of the bad events occur. But then $\phi$ is a near-induced minor model of $\dotdot{H}$ in $G$, which means that $G$ contains $H$ as an induced minor, completing the proof. 

We note that the idea of thinking of concurrent flows as a probability distribution and sampling from this distribution to  obtain a model of a graph $H$ in $G$ has been used before, see e.g.~\cite{DBLP:journals/corr/abs-1901-09349,DBLP:journals/jct/GroheM09,DBLP:journals/dmtcs/HatzelKPS22,DBLP:conf/soda/JaffkeLMPS23,marx-toc-treewidth} in the context of minors. On the other hand we we are not aware of previous approaches doing this for {\em near}-induced minors, and shooting for $\dotdot{H}$ as opposed to $H$. This appears crucial for obtaining an {\em induced} minor model of $H$ (as opposed to just a minor model).

\subsection{Overview of Subexponential Time Algorithms}\label{sec:overAlgo}
We first sketch how \Cref{thm:septheorem} implies an algorithm with running time $n^{\OO_H(n^{2/3} \log n)}$ for the {\sc Maximum Independent Set} problem on $H$-induced-minor-free graphs.
Here input is an $H$-induced-minor-free graph $G$ on $n$ vertices and the task is to find a maximum size set $S$ such that no edge has both endpoints in $S$. Our algorithm is a recursive branching algorithm that takes as input a graph $G$ and set $I$, and finds a largest independent set $S$ in $G$ such that $I \subseteq S$.
As long as there is a vertex $v \notin I$ of degree at least $n^{1/3}$ in $G \setminus I$ the algorithm branches on $v$. This means that it finds the best solution not containing $v$ by calling itself recursively on $(G \setminus v, I)$ and the best solution containing $v$ by calling itself recursively on $(G \setminus N(v), I \cup \{v\})$. Each time we include $v$ in $I$ the vertex set of $G$ decreases by at least $n^{1/3}$ vertices. The running time of the algorithm is governed by the recurrence $T(n) \leq T(n-1) + T(n - n^{1/3})$ which solves to $T(n) \leq n^{\OO(n^{2/3})}$.
Thus the recursion above produces at most  $n^{\OO(n^{2/3})}$ instances of the form $(G, I)$ where $G \setminus I$ has maximum degree  $n^{1/3}$. We now show how to solve each of those instances in time  $n^{O_H(n^{2/3})}$, leading to a total time of  $n^{O_H(n^{2/3})}$ for the algorithm.
Before we proceed we observe that all of the instances have one additional property, namely that the size of the set $I$ is at most  $n^{2/3}$. Indeed, each time the algorithm adds a vertex to $I$ it decreases $V(G)$ by $n^{1/3}$, and this cannot happen more than $n^{2/3}$ times.

We now handle the bounded maximum degree case. 
Since $G \setminus I$ has maximum degree $n^{1/3}$ then by \Cref{thm:septheorem} it has a balanced separator of $Z$ size at most $O_H(n^{2/3})$. 
Each of the connected components of $G \setminus (I \cup Z)$ have at most $2n/3$ vertices and also have maximum degree of size at most $n^{1/3}$. So by \Cref{thm:septheorem} each connected component  of $G \setminus (I \cup Z)$ has a balanced separator of size at most $O_H(\left(\frac{2n}{3}\right)^{2/3})$. 
Recursively applying \Cref{thm:septheorem} in this manner shows that $G \setminus I$ has treewidth at most  $O_H(n^{2/3})$ (see e.g.~\cite{DBLP:journals/jct/DvorakN19}). 

For readers unfamiliar with treewidth, the only facts about it that we will use are that (i) a number of problems (including {\sc Independent Set}) can be solved in time $2^{\OO(t \log t)}n$ or $2^{\OO(t)}n$ on graphs of treewidth $t$, and that adding $k$ vertices to a graph of treewidth $t$ increases treewidth by at most $k$.
Since $G \setminus I$ has treewidth at most  $O_H(n^{2/3})$  and $|I| \leq n^{2/3}$ the treewidth of $G$ is also at most $O_H(n^{2/3})$. We now use the  $2^{\OO(t)}n$ time algorithm for {\sc Independent Set} on graphs of treewidth $t$ and obtain an algorithm with the desired running time for {\sc Independent Set} on $H$-induced-minor-free graphs.

We now adapt the approach above to the {\sc Maximum Induced Forest} problem (which is equivalent to  {\sc Feedback Vertex Set}).
Here input is a graph $G$ on $n$ vertices and the task is to find a largest possible set $S$ such that $G[S]$ is acyclic
Acyclic graphs are also called {\em forests}. 
We still want to branch on high degree vertices to produce $n^{\OO(n^{2/3})}$ instances of the form $(G, I)$ where $G \setminus I$ has maximum degree  $n^{1/3}$ and $|I| = \OO(n^{2/3})$. The exact same argument then yields that the treewidth of $G$ is at most  $O_H(n^{2/3})$, and we may apply the $2^{\OO(t \log t)}n$ (or even $2^{\OO(t)}n^{\OO(1)}$) time~\cite{cygan2015parameterized,cut-and-count,rank-treewidth,FominLS14} algorithms for  {\sc Maximum Induced Forest} on graphs of treewidth $t$ to obtain the desired results. 
The only difficulty with executing this plan lies in the branching algorithm. In {\sc Independent Set} when we included a vertex $v$ in $I$ we could immediately conclude that none of $v$'s neighbors could be in the solution, so we could delete them. This is not the case for {\sc Maximum Induced Forest}. Nevertheless we can do something similar, and show that including a vertex in $I$ indirectly forces us (perhaps after a few additional branching steps) to delete many vertices from $G$. 

The key observation is that forests are $1$-{\em degenerate}: for every forest $G$ there exists an ordering $\eta : V(G) \rightarrow [|V(G)|]$ so that every vertex has at most $1$ neighbor $u$ with $\eta(u) < \eta(v)$. We will call $u$ the {\em left neighbor} of $v$.
Now the trick (inspired by  Gartland et al.~\cite{DBLP:conf/stoc/GartlandLPPR21})
is the following: we give each vertex of $G$ a coin. When we include a vertex $v$ in $I$ we branch in $|N(v)| + 1$ many ways, in particular the algorithm guesses which neighbor of $v$ is the left neighbor of $v$ in the solution $G[S]$ (the algorithm also considers the possibility that $v$ has no left neighbors in $G[S]$).
For each neighbor $w$ of $v$ that is not guessed to be $v$'s left neighbor we take away its coin. Note that if $w$ is in the solution forest $S$ then $v$ must be a left neighbor of $w$. Thus, if $w$ does not have a coin to pay with when $v$ is inserted in $I$ then $w$ must have {\em two} left neighbors in $S$ which is impossible. So $w$ can't be part of $S$ and we therefore delete $w$ from $G$. 

The observation above implies that when the algorithm inserts a vertex $v$ into $I$ each neighbor of $v$ (except at most one) either loses its coin or disappears from $G$ entirely. We branch on a vertex $v$ of degree at least $n^{1/3}$ and track the progress of the algorithm by the measure $\mu$, which is defined to be the number of vertices plus the number of coins. The number of leaves of the recursion tree as a function of $\mu$ is then governed by the following recurrence. 
$$ T(\mu) \leq T(\mu - 1) + (|N(v)|+1) \cdot T(\mu - n^{1/3}) \leq T(\mu - 1) + n \cdot T(\mu - n^{1/3})$$
Here the $(|N(v)|+1)$ factor came from guessing which vertex $u$ is the left neighbor of $v$. This recurrence solves to $T(\mu) \leq n^{O\left(\frac{\mu}{n^{1/3}}\right)}$. Because $\mu$ starts out with value $2n$ (each vertex has one coin) this leads to the desired $n^{\OO(n^{2/3})}$ bound. Since the branching step was the only place where we needed to change something from the algorithm for {\sc Independent Set}, this yields a  $n^{O_H(n^{2/3})}$ time algorithm for {\sc Maximum Induced Forest}.


The algorithms for the other problems of \Cref{cor:algos} closely follow the algorithm for {\sc Maximum Induced Forest}. In particular all of the problems admit $2^{\OO(t \log t)}n^{\OO(1)}$ time algorithms on graphs of treewidth at most $t$, and all of the problems (with the exception of 3-{\sc Coloring}, which is handled using some additional ideas from~\cite{DBLP:journals/algorithmica/BonnetR19}) ask for finding a maximum induced $\delta$-degenerate subgraph with some additional properties. The algorithm and analysis for maximum induced $\delta$-degenerate subgraph is almost the same as for  {\sc Maximum Induced Forest}, the difference is that every vertex gets $\delta$ coins instead of just one, and that when $v$ is inserted into $I$ we need to guess a set of $\delta$ left neighbors of $v$ who do not need to pay. The algorithm of \Cref{cor:indminregoc} follows the exact same scheme. Indeed, for every $H$ where every edge is incident to a vertex of degree at most $2$, every minimal induced minor model of $H$ is $\delta$-degenerate for some constant $\delta$ depending on $H$.

\subsection{Overview of Hardness of Induced Minor Testing}\label{sec:overHardness}
We give an overview of the proof of \Cref{thm:hardnesstheorem}. In particular we show hardness (in the sense of both NP-hardness and the non-existence of $2^{o(n/\log^3 n)}$ time algorithms assuming the ETH) of $T$-{\sc Induced Minor Testing} for a fixed tree $T$.
We will reduce from $3$-{\sc Coloring}. This problem is known to be NP-complete, and not to admit an algorithm with running time $2^{o(|E(G)|)}$ assuming the ETH~\cite{ImpagliazzoPZ01}. We go through the following chain of reductions:

\smallskip
\makebox[3cm]{$3$-{\sc Coloring}}\makebox[1cm]{$\leq$} {\sc Generalized $3$-Coloring on Binary Shift Graphs}\par
\makebox[3cm]{}\makebox[1cm]{$\leq$} {\sc Multicolored Induced $6$-Disjoint Paths}\par
\makebox[3cm]{}\makebox[1cm]{$\leq$} $18$-{\sc Anchored} $T^\star$-{\sc Induced Minor Testing}\par
\makebox[3cm]{}\makebox[1cm]{$\leq$} $T$-{\sc Induced Minor Testing}\par
\smallskip
We will describe each step of the reduction one by one in (almost) reverse order, working our way backwards from $T$-{\sc Induced Minor Testing} and all the way to $3$-{\sc Coloring}.

\paragraph{Anchored Minor Testing to Minor Testing.}
An {\em induced minor model} of $H$ in $G$ is a labeled collection $\{X_v \mid v \in V(H)\}$ of pairwise disjoint vertex subsets $X_v \subseteq V(G)$ of $V(G)$, so that {\em (i)}  $G[X_v]$ is connected for all $v \in V(H)$, and {\em (ii)} for distinct $u,v \in V(H)$, the sets $X_u$ and $X_v$ are adjacent in $G$ if and only if $uv \in E(H)$. For $v \in V(H)$ the set $X_v$ is called the {\em branch set} of $v$. For each integer $h \geq 0$ we will denote by $B_h$ the complete binary tree of height $h$ (the $B_0$ is the graph on one vertex).
For a constant integer $a$ and graph $H$, the $a$-{\sc Anchored} $H$-{\sc Induced Minor Testing} problem is defined as follows.
The input consists of a graph $G$ and a list of $a$ anchors $(v_1,u_1), \ldots, (v_a, u_a)$ that are pairs $(v_i,u_i) \in V(G) \times V(H)$. \todo{technically only the list of vertices in $G$ should be part of the input?}
The problem is to decide whether $G$ contains an induced minor model of $H$, so that for every anchor $(v_i,u_i)$ it holds that the branch set of $u_i \in V(H)$ contains the vertex $v_i$ of $G$. 

In the last step of our chain of reductions we reduce {\em from} the $18$-{\sc Anchored} $T^\star$-{\sc Induced Minor Testing} problem on $B_{285}$-induced minor free graphs. Here the tree $T^\star$ is the tree depicted in Figure~\ref{fig:treeOne}. 
\begin{figure}[htb]
\begin{center}
\includegraphics[width=0.6\textwidth]{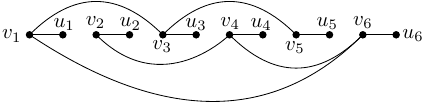}
\caption{The construction of $T^\star$.}
\label{fig:treeOne}
\end{center}
\end{figure}
The tree $T$ (for which we will show hardness of $T$-{\sc Induced Minor Testing}) is obtained from $T^\star$ by attaching for each of the $18$ anchor pairs $(v_i, u_i)$ (with $v_i \in V(G)$ and $u_i \in V(T^\star)$) a copy of $B_{285+2i}$ to $u_i$.
To reduce from $18$-{\sc Anchored} $T^\star$-{\sc Induced Minor Testing} (in a $B_{285}$-induced minor free graph $G$) to $T$-{\sc Induced Minor Testing} we attach for each of the $18$ anchor pairs $(v_i, u_i)$ a copy of $B_{285+2i}$ to $v_i$. Let $G'$ be the resulting graph. 
It is now not too difficult to convince oneself that any induced minor model of $T$ in $G'$ must place the $i$'th ``attached'' $B_{285+2i}$ in the copy of $B_{285+2i}$ attached to $v_i$ in $G$, and that therefore such an induced minor model of $T$ in $G'$ contains an induced minor model of $T^\star$ in $G$ that respects the anchors. 

The idea of first showing hardness of $a$-{\sc Anchored} $H$-{\sc Induced Minor Testing} on graphs that exclude some fixed graph $H^\star$ as an induced minor, and then attach (constant size) graphs that {\em do} contain $H^\star$ (both to $H$ and to the input graph $G$) to force the model of $H$ to respect the anchors comes from the reduction of Fellows et al.~\cite{DBLP:journals/algorithmica/FellowsKMP95}. 
They first show that $a$-{\sc Anchored} $H$-{\sc Induced Minor Testing} remains NP-hard for some fixed graph $H$, when the input graph $G$ is required to be planar. Setting $H^\star$ to be the complete graph on $5$ vertices we observe that planar graph  exclude $H^\star$ as a {\em minor}, and therefore also as an {\em induced} minor. 

This is where the similarities between our reduction and the reduction of Fellows et al.~\cite{DBLP:journals/algorithmica/FellowsKMP95} end. In particular we need to show hardness of $a$-{\sc Anchored} $T^\star$-{\sc Induced Minor Testing} when the input graph $G$ excludes some fixed tree (in our case $B_{285}$) as an induced minor. We cannot follow the approach of Fellows et al. and focus our attention on graphs that exclude $B_{285}$ as a {\em minor}, because such graphs have bounded pathwidth and therefore, by Courcelle's Theorem~\cite{DBLP:journals/iandc/Courcelle90},  $a$-{\sc Anchored} $T^\star$-{\sc Induced Minor Testing} can be solved in linear time on such graphs.

\paragraph{Induced Disjoint Paths to Anchored Minor Testing}
To show hardness of $18$-{\sc Anchored} $T^\star$-{\sc Induced Minor Testing} on  $B_{285}$-induced minor free graphs we first show hardness of the {\sc Multicolored Induced $6$-Disjoint Paths} problem. 
Here input consists of a graph $G$, a partition $V_1, \ldots, V_6$ of $V(G)$ into $6$ parts, and $6$ pairs of terminal vertices $(s_1, t_1), \ldots, (s_6, t_6)$, with $s_i, t_i \in V_i$.
The problem is to decide whether there exists $6$ paths $P_1, \ldots, P_6$, so that for each $i \in [6]$, $P_i$ is a $s_i-t_i$-path that is contained in $V_i$, and there are no edges between $P_i$ and $P_j$ for every $i \neq j$.
We will only consider instances of  {\sc Multicolored Induced $6$-Disjoint Paths} that satisfy that there are no edges between $V_i$ and $V_j$ when $|i-j|>1$.

It is not too difficult to reduce from {\sc Multicolored Induced $6$-Disjoint Paths} to $18$-{\sc Anchored} $T^\star$-{\sc Induced Minor Testing}.
Given as input an instance $G$, $V_1, \ldots V_6$, and $(s_1, t_t), \ldots,  (s_6, t_6)$ of  {\sc Multicolored Induced $6$-Disjoint Paths} we construct 
an instance of  $18$-{\sc Anchored} $T^\star$-{\sc Induced Minor Testing} consisting of a graph $G'$ together with $18$ anchors, as follows. 
The graph $G'$ is obtained from $G$ by adding for every $i \in [6]$ a vertex $w_i$ that is universal to $V_i$ and non-adjacent to everything else. Additionally, for every pair $i,j$ so that there is an edge between $v_i$ and $v_j$ in $T^\star$, we add all possible edges between $V_i$ and $V_j$. This completes the construction of $G'$.
Finally, for every $i \leq 6$, we label $w_i$ with an anchor for $u_i$ and both $s_i$ and $t_i$ with an anchor for $v_i$.
The proof that the two instances are equivalent is fairly standard and we therefore skip it in this overview. We do remark that this proof crucially relies on the assumption that in $G$ there are no edges between $V_i$ and $V_j$ when $|i-j|>1$.

The interesting part of the reduction is how to ensure that the produced graph $G'$ excludes a $B_{285}$ as an induced minor. 
Since the reduction does not modify $G[V_i]$ for each $i \in [6]$ it is clear that we need to ensure that the instance $G$, $V_1, \ldots V_6$, and $(s_1, t_t), \ldots,  (s_6, t_6)$  that we reduce {\em from} at the very least satisfies that $G[V_i]$ excludes  $B_{285}$ as an induced minor for every $i \in [6]$.
It turns out a slightly stronger condition is not only necessary, but also sufficient. 
In particular it suffices to show hardness of {\sc Multicolored Induced $6$-Disjoint Paths} such that for every $i \leq 5$ it holds that $G[V_i \cup V_{i+1}]$ excludes a $B_{280}$ as an induced minor (recall that additionally we assume that there are no edges between $V_i$ and $V_j$ when $|i-j|>1$).

Indeed, suppose for contradiction that $G$ has this property but that $G'$ contains an induced minor model of $B_{285}$. 
For $i \in [6]$, let us call $V_i$ important if at least three branch sets of the induced minor model of $B_{285}$ intersect $V_i$, and unimportant otherwise.
Observe that if there is an edge between $v_i$ and $v_j$ in $T^\star$, then at most one of $V_i$ and $V_j$ is important, as otherwise $B_{285}$ would contain a cycle.
We delete from the induced minor model of $B_{285}$ all branch sets that intersect an unimportant set $V_i$, and all branch sets that intersect a $w_i$ vertex.
We deleted at most $18$ branch sets, but $B_{285}$ contains $32$ vertex disjoint copies of $B_{280}$. Pick one copy which did not contain any of the removed branch sets. 

Let $i \in [6]$ be the smallest integer so that the induced minor model of $B_{280}$ intersects $V_i$.
We claim that the induced minor model of $B_{280}$ is contained in $V_i \cup V_{i+1}$.
First, because the sets $V_j$ that intersect the model are important, they must correspond to an independent set of $T^\star$.
Therefore, we observe that the only edges of $G'$ that the induced minor model can use are the edges of $G$, and therefore the sets $V_i, \ldots, V_j$ that the model intersects must be consecutive because $B_{280}$ is connected.
However, the model cannot intersect $V_{i+2}$, so the sets are at most $V_i$ and $V_{i+1}$.
But now $B_{280}$ is an induced minor of $G'[V_i \cup V_{i+1}] = G[V_i \cup V_{i+1}]$, contradicting the assumption that $G[V_i \cup V_{i+1}]$ excludes $B_{280}$ as an induced minor.

\paragraph{From $3$-Coloring to Generalized $3$-Coloring on Binary Shift Graphs}
It remains to show hardness of  {\sc Multicolored Induced $6$-Disjoint Paths} for graphs that satisfy that there are no edges between $V_i$ and $V_j$ when $|i-j|>1$, and for every $i \leq 5$ it holds that $G[V_i \cup V_{i+1}]$ excludes a $B_{280}$ as an induced minor. We reduce from $3$-{\sc Coloring} to {\sc Multicolored Induced $6$-Disjoint Paths} via
an intermediate problem, called {\sc Generalized $3$-Coloring} on a class of graphs which admits a special edge partition.
We now discuss the reduction from $3$-{\sc Coloring} to  {\sc Generalized $3$-Coloring}

%
In {\sc Generalized $3$-Coloring} input is a graph whose edges have been colored red or green. The task is to determine whether there exists a mapping $c : V(G) \rightarrow \{1,2,3\}$ such that for every red edge $uv$ it holds that $c(u) \neq c(v)$ and for every green edge it holds that  $c(u) = c(v)$. The green edges encoding ``equality constraints'' allow us to reduce from $3$-{\sc Coloring} of a graph $G$ to {\sc Generalized $3$-Coloring} of any graph $G'$ which contains $G$ as a minor. This allows for a lot of flexibility in the choice of $G'$ which in turn is very useful when reducing to  {\sc Multicolored Induced $6$-Disjoint Paths}.

We will select the graph $G'$ to be a binary shift graph. For integer $b \ge 1$, the binary shift graph $\BS_b$ has $2^b$ vertices $v_0, \ldots, v_{2^b-1}$.
The binary shift graph $\BS_b$ has an edge between vertices $v_x$ and $v_y$ if $x \neq y$ and either
\begin{align*}
x \equiv 2y \mod{2^b}, ~~~ x \equiv 2y+1 \mod{2^b}, ~~~ y \equiv 2x \mod{2^b}, ~~~\mbox{ or } y \equiv 2x+1 \mod{2^b}\mbox{.}
\end{align*}
In particular, $v_x$ is adjacent to $v_y$ if the length-$b$ binary representation of $y$ can be obtained from the length-$b$ binary representation of $x$ by ``shifting'' it by one digit to left or right and setting the new digit to be either $0$ or $1$. It is worth noting that Binary Shift graphs are also known as {\em undirected De-Bruijn graphs}, and have been extensively studied (see e.g.~\cite{bermond1989bruijn,delorme1998spectrum,pradhan1991fault}).
It is already known (see e.g.~\cite{delorme1998spectrum}) that for every $b$,  $\BS_b$ does not admit a balanced separator of size $o(2^b/b)$. But, as we saw from \Cref{thm:septheorem}, the non-existence of sufficiently sub-linear size balanced separators implies that $\BS_b$ contains every sufficiently small graph as an induced minor! Since we want to explicitly compute the minor model of $G$ in $\BS_b$ (to use in the reduction), instead of using \Cref{thm:septheorem} we will use an analogous result for minors by Krivelevich and Nenadov~\cite{DBLP:conf/bcc/Krivelevich19} which implies that there exists a $c > 0$ such that for every graph $G$ such that $|V(G)|+|E(G)| \le c \cdot 2^b/b^3$ we can compute a minor model of $G$ in $\BS_b$ in polynomial time. 
Observe that $|V(\BS_b)| = 2^b$. Since we can embed arbitrary $m$-edge $3$-{\sc Coloring} instances into {\sc Generalized $3$-Coloring} on $\BS_b$, where $m = \OO(2^b/b^3)$ it follows that, assuming the ETH, {\sc Generalized $3$-Coloring} on Binary Shift graphs does not have an algorithm with running time $2^{o(n/\log^3 n)}$.

\paragraph{From Generalized $3$-Coloring to Multicolored Induced $6$-Disjoint Paths}
Finally we reduce from {\sc Generalized $3$-Coloring} on binary shift graphs to {\sc Multicolored Induced $6$-Disjoint Paths}. We are given a binary shift graph $\BS_b$ with edges colored green and red. 
Our aim is to produce an equivalent instance $G$, $V_1, \ldots V_6$, and $(s_1, t_t), \ldots,  (s_6, t_6)$ of  {\sc Multicolored Induced $6$-Disjoint Paths} such that there are no edges between $V_i$ and $V_j$ for $|i-j|>1$ and such that $G[V_i \cup V_{i+1}]$ excludes a $B_{280}$ as an induced minor for every $i \leq 5$. 
The graph $G$ that we construct will satisfy a stronger property - that $G[V_i \cup V_{i+1}]$ excludes a $B_{280}$ as a minor. In particular we will construct instances where the pathwidth of  $G[V_i \cup V_{i+1}]$ is upper bounded by $139$. It is well known that graphs of pathwidth $h$ exclude the $B_{2(h+1)}$ as a minor~\cite{scheffler1989baumweite}. 

This is the place where we use the special structure of binary shift graphs. Recall that the vertices of $\BS_b$ are numbered as $v_0, \ldots v_{2^b-1}$. We define $P$ to be the edge set of a path that visits all the vertices in this order (this path is not a path in $\BS_b$). Specifically $P = \{v_iv_{i+1} ~:~ i < 2^b\}$. It turns our that the edges of $\BS_b$ can be partitioned into $5$  sets $E_1$, $E_2$, $E_3$, $E_4$, $E_5$ such that for every $i \leq 5$ the pathwidth of the graph $(V(\BS_b), P \cup E_i)$ is at most $16$ (this proof is not too complicated, but a little bit technical so we skip it in this overview, see \Cref{lem:binshiftpartition}). 

We are now ready to carry out the final construction. For the remainder of the discussion, let $n = 2^b$ be the number of vertices of the graph $\BS_b$ in the instance we reduce from. We first construct the graph $G[V_1]$ as shown in \Cref{fig:gv1over}. The solution path in $G[V_1]$ should go from $s_1$ to $t_1$. For each $j \leq n$ the path from $s_1$ to $t_1$ passes through exactly one of $a_1^j$, $b_1^j$, or $c_1^j$. Which vertex the path passes through encodes the choice of color of the vertex $v_j$ in the {\sc Generalized $3$-Coloring} instance (say $a_1^j$ means $c(v_j) = 1$, $b_1^j$ means $c(v_j) = 2$ and $c_1^j$ means $c(v_j) = 3$).
For each $2 \leq i \leq 6$ the graph $G[V_i]$ is a copy of $G[V_1]$. We will use the subscripts to denote which copy a vertex belongs to (so $a_5^3$ is the vertex in $V_5$ that corresponds to coloring $v_3$ with color $1$).

\begin{figure}[htb]
\begin{center}
\includegraphics[width=0.9\textwidth]{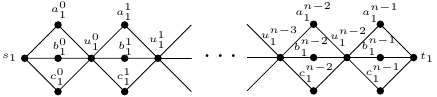}
\caption{The construction of $G[V_1]$.}
\label{fig:gv1over}
\end{center}
\end{figure}

We want each of the $6$ paths to encode the same coloring of the vertex set of $\BS_b$. For this we add edges $\{a_i^jb_{i+1}^j, a_i^jc_{i+1}^j, b_i^ja_{i+1}^j, b_i^jc_{i+1}^j, c_i^ja_{i+1}^j, c_i^jb_{i+1}^j\}$ for every $i \leq 5$ and $j \leq n$. Since the path from $s_i$ to $t_i$ should have no edges to the path from  $s_{i+1}$ to $t_{i+1}$ this ensures that the $6$ paths do the exact same thing (that is, correspond to the exact same $c : V(\BS_b) \rightarrow \{1,2,3\}$.)

Finally we add edges to encode the edges of $\BS_b$. For each $i \leq 5$ we add the edges corresponding to $E_i$ (recall that bounded pathwidth partition!) between $G[V_i]$ and  $G[V_{i+1}]$. So each part $E_i$ of the partition corresponds to one consecutive $V_i, V_{i+1}$ pair in $G$.
For a red edge $v_{j}v_{j'}$ we add the edges $\{a_i^ja_{i+1}^{j'}, b_i^jb_{i+1}^{j'}, c_i^jc_{i+1}^{j'}\}$ to $G$. For a green edge  $v_{j}v_{j'}$ we add the edges  $\{a_i^jb_{i+1}^{j'}, a_i^jc_{i+1}^{j'}, b_i^ja_{i+1}^{j'}, b_i^jc_{i+1}^{j'}, c_i^ja_{i+1}^{j'}, c_i^jb_{i+1}^{j'}\}$. It is not too hard to verify that these edges encode precisely the inequality and equality constraints enforced by the red and green edges respectively. 

It remains to show that the pathwidth of $G[V_i \cup V_{i+1}]$ is bounded by $139$ for every $i \leq 5$.
To see this take a path decomposition of $(V(\BS_b), P \cup E_i)$ (which had width at most $16$) and replace every occurrence of every vertex $v_j$ in a bag with the vertex set $\{a^j_i,b^j_i,c^j_i,u^j_i, a^j_{i+1},b^j_{i+1},c^j_{i+1},u^j_{i+1}\}$ (and finally add $s_i,t_i,s_{i+1},t_{i+1}$ to all bags).
The replacement increases the width of the path decomposition by a factor of at most $8$, while the addition increases it by at most $4$. So the pathwidth of  $G[V_i \cup V_{i+1}]$ is at most $(16+1) \cdot 8 + 4 - 1 = 139$. 
%
We have now sketched each of the four steps of the reduction of \Cref{thm:hardnesstheorem}. We note that the $1/\log^3 n$ factor in the exponent is entirely due to the reduction from $3$-{\sc Coloring} to {\sc Generalized $3$-Coloring} on Binary Shift graphs, all of the other steps of the reduction only lose constant factors in the exponent. 



\section{Preliminaries}
\label{sec:prelims}
For integer $n$, we denote by $[n] = \{1,\ldots,n\}$ the set of positive integers that are at most $n$.
For integers $a$ and $b$, we denote by $[a,b] = \{a,\ldots, b\}$ the set of integers that are at least $a$ and at most $b$.
We use $\log$ to denote base-2 logarithm.
We usually denote an unordered pair of $a$ and $b$ by $ab$ and an ordered pair by $(a,b)$.

\subsection{Graphs}
All graphs in this paper are undirected and simple.
The set of vertices of a graph $G$ is denoted by $V(G)$ and the set of edges by $E(G)$.
The set of neighbors of a vertex $v$ in $G$ is denoted by $N_G(v)$, or by $N(v)$ if the graph $G$ is clear from the context.
For a set of vertices $X$, the subgraph induced by $X$ is denoted by $G[X]$.
We also use $G \setminus X$ to denote $G[V(G) \setminus X]$.
We denote by $\onedot{G}$ the graph obtained from $G$ by subdividing every edge once, and by $\dotdot{G}$ the graph obtained from $G$ by subdividing every edge twice.

An induced minor model of a graph $H$ in a graph $G$ is a labeled collection $\{X_v \mid v \in V(H)\}$ of pairwise disjoint vertex subsets $X_v \subseteq V(G)$ of $V(G)$, so that
\begin{enumerate}
\item $G[X_v]$ is connected for all $v \in V(H)$, and
\item for distinct $u,v \in V(H)$, the sets $X_u$ and $X_v$ are adjacent in $G$ if and only if $uv \in E(H)$.
\end{enumerate}

A graph $G$ contains $H$ as an induced minor if there is an induced minor model of $H$ in $G$.
Equivalently, $G$ contains $H$ as an induced minor if $H$ can be obtained from $G$ by vertex deletions and edge contractions.

A tree decomposition of a graph $G$ is a pair $(T,\bag)$, where $T$ is a tree and $\bag \colon V(T) \rightarrow 2^{V(G)}$ is a function mapping each node of $T$ into a subset of $V(G)$ called a bag, so that
\begin{enumerate}
\item for every $uv \in E(G)$, there exists $t \in V(T)$ so that $\{u,v\} \subseteq \bag(t)$, and
\item for every $v \in V(G)$, the subtree of $T$ induced by $\{t \in V(T) \mid v \in \bag(t)\}$ is non-empty and connected.
\end{enumerate}
The width of a tree decomposition is the maximum size of a bag minus $1$, and the treewidth of a graph is the minimum width of a tree decomposition of it.
A path decomposition is a tree decomposition where the tree $T$ is a path, and pathwidth is defined analogously to treewidth but with path decompositions.

The complete binary tree $B_h$ of height $h$ is defined recursively by letting $B_1$ to be a tree with a single vertex, and then for $h \ge 2$ the tree $B_h$ to be obtained from $B_{h-1}$ by adding two nodes of degree $1$ adjacent to each node of $B_{h-1}$ of degree at most $1$.
For all $h \ge 1$, a complete binary tree $B_{2h}$ of height $2h$ has pathwidth at least $h$~\cite{scheffler1989baumweite}.

A graph $G$ is $\delta$-degenerate if there exists a bijective function $\eta \colon V(G) \rightarrow [|V(G)|]$ so that every $v \in V(G)$ has at most $\delta$ neighbors $u \in N(v)$ with $\eta(u) < \eta(v)$.
Such function $\eta$ is called a $\delta$-degeneracy ordering of $G$.
We say that a vertex $u$ is a left neighbor of $v$ with respect to $\eta$ if $u \in N(v)$ and $\eta(u) < \eta(v)$.

\subsection{Separations and concurrent flows}
A triple of disjoint subsets of vertices $(A,S,B)$ is a separation in $G$ if $A \cup S \cup B = V(G)$ and there are no edges between $A$ and $B$.
The order of a separation is $|S|$, and we say that a separation is \emph{balanced} if $\max(|A|, |B|) \le \frac{2}{3} \cdot |V(G)|$.
We say that a set $S \subseteq V(G)$ is a \emph{balanced separator} if there exists a balanced separation $(A,S,B)$, or equivalently, every connected component of $G \setminus S$ has size at most $\frac{2}{3} \cdot |V(G)|$.
The \emph{sparsity} of a separation $(A,S,B)$ is

\begin{align*}  
\alpha(A,S,B) = \frac{|S|}{|A \cup S| \cdot |B \cup S|}.
\end{align*}

We use the convention that a path $P$ in a graph $G$ is a sequence of pairwise distinct vertices $v_1, \ldots, v_{\ell}$ with $\ell \ge 1$ so that $v_i$ is adjacent to $v_{i+1}$ for all $i<\ell$.
We say that $P$ is an $s-t$-path if $v_1 = s$ and $v_{\ell} = t$.
We denote by $V(P)$ the set of vertices in a path $P$.
We use $\paths(G)$ to denote the set of all paths in $G$, and $\paths_{(s,t)}(G)$ to denote the set of all $s-t$-paths in $G$.

A \emph{concurrent flow} in a graph $G$ is a function $\cflow \colon \paths(G) \rightarrow \mathbb{R}_{\ge 0}$ that maps paths of $G$ to non-negative reals so that for every ordered pair of vertices $(a,b) \in V(G)^2$, it holds that $\sum_{P \in \paths_{(a,b)}(G)} \cflow(P) = 1$.
The \emph{congestion} of a concurrent flow is the maximum amount of flow going through a single vertex, i.e., $\max_{v \in V(G)} \sum_{P \in \{P \in \paths(G) \mid v \in V(P)\}} \cflow(P)$.

We will use the following theorem of~\cite{FeigeHL08,DBLP:journals/jacm/LeightonR99} that relates the minimum congestion of concurrent flow with the minimum sparsity of a separation.

\begin{proposition}[\cite{FeigeHL08,DBLP:journals/jacm/LeightonR99}]
\label{pro:leightrao}
There is a polynomial-time algorithm that given a graph $G$ and a number $\cng$, outputs either a concurrent flow of congestion at most $\cng$ or a separation with sparsity $\OO(\log n/\cng)$.
\end{proposition}

We note that we can assume that only a polynomial number of paths in a concurrent flow have a non-zero flow~\cite{DBLP:journals/jacm/EdmondsK72}, so a concurrent flow has a polynomial-size representation (up to a desired precision of the real numbers).

For improving the separator theorem for string graphs from $\OO(\sqrt{m} \log m)$ to $\OO(\sqrt{m})$, Lee~\cite{DBLP:conf/innovations/Lee17} showed that the $\log n$ in \Cref{pro:leightrao} could be replaced by $h^2$ in graphs that exclude $\onedot{K_h}$ as an induced minor.

\begin{proposition}[\cite{DBLP:conf/innovations/Lee17}]
\label{pro:leekpr}
There is a randomized polynomial-time algorithm that given a graph $G$, a number $\cng$, and an integer $h$, outputs either a concurrent flow of congestion $\cng$, a separation with sparsity $\OO(h^2/\cng)$, or an induced minor model of $\onedot{K_h}$ in $G$.
\end{proposition}

\subsection{Lov{\'a}sz local lemma}
The Lov{\'a}sz local lemma~\cite{erdos1975problems} states the following.

\begin{proposition}[\cite{erdos1975problems}]
Let $\events$ be a set of random events such that each event $A \in \events$ occurs with probability at most $p$ and each event is independent of all but at most $d$ of the other events.
If $p \le 1/(4d)$, then there is a non-zero probability that none of the events occur.
\end{proposition}

We will use a constructive version of the local lemma.
This was given by Moser and Tardos~\cite{DBLP:journals/jacm/MoserT10}.
Next we state their result in a form that is sufficient for our purposes.
Let $\rvars$ be a finite set of mutually independent random variables and $\events$ a finite set of events that depend on $\rvars$ (i.e whether or not each event in $\events$ occurs is a function of the $\rvars$.)
For an event $A \in \events$, let $\vbl(A) \subseteq \rvars$ be a subset of the variables so that whether $A$ occurs is determined by $\vbl(A)$.
We say that an event $A \in \events$ is \emph{adjacent} to an element $B \in \events$ if $A \neq B$ and $\vbl(A) \cap \vbl(B) \neq \emptyset$.
The \emph{degree} of $A \in \events$ is the number of events that are adjacent to $A$.

\begin{proposition}[\cite{DBLP:journals/jacm/MoserT10}]
\label{pro:lllconstr}
There is a polynomial-time algorithm that takes as an input
\begin{itemize}
\item a set $\rvars$ of mutually independent random variables that each can be sampled in polynomial-time and
\item a set of events $\events$ and a mapping $\vbl \colon \events \rightarrow 2^{\rvars}$ so that whether $A \in \events$ occurs is determined by $\vbl(A)$ and can be computed in polynomial-time.
\end{itemize}
If the degree of each event $A \in \events$ is at most $d$ and the probability is at most $\Pr[A] \le 1/(4d)$, then the algorithm finds an assignment of $\rvars$ so that none of the events occurs.
\end{proposition}

\section{Separator Theorem}
\label{sec:septheorem}
In this section we prove \Cref{thm:septheorem}.
The general structure of the proof is that we first construct a balanced separator or a concurrent flow, then from the concurrent flow we construct so-called ``induced almost-embedding'', and then from an induced almost-embedding we construct and induced minor model.

\subsection{Concurrent flow}
We start by putting \Cref{pro:leightrao,pro:leekpr} together and formulating them as to return a balanced separation instead of a separation with bounded sparsity.
This lemma is standard material (e.g.~\cite{FeigeHL08}), but we state and prove it in our notation for the readers convenience.

\begin{lemma}
\label{lem:obtconflow}
There is a randomized polynomial-time algorithm that given graphs $G$ and $H$, and a number $\cng$, outputs one of the following:
\begin{enumerate}
\item\label{enum:ret:conflow} an induced subgraph $G'$ of $G$ with at least $\frac{2}{3} \cdot |V(G)|$ vertices together with a concurrent flow of congestion $\cng$ in $G'$, or
\item\label{enum:ret:minor} a balanced separation of $G$ of order $\min(\log |V(G)|, |V(H)|^2) \cdot \OO(|V(G)|^2/\cng)$, or
\item an induced minor model of $H$ in $G$.
\end{enumerate}
\end{lemma}
\begin{proof}
We will apply \Cref{pro:leightrao} or \Cref{pro:leekpr} repeatedly, depending on whether $\log n$ or $|V(H)|^2$ is smaller.
In particular, \Cref{pro:leightrao} is applied if $\log n$ is smaller than $|V(H)|^2$, and otherwise \Cref{pro:leekpr} is applied with $h = |V(H)|$.
For convenience, let us denote $\mu = \min(\log n, |V(H)|^2)$.

We do a process that iterates through a series of separations $(A_0,S_0,B_0), (A_1,S_1,B_1) \ldots$.
We initially set $(A_0,S_0,B_0) = (V(G),\emptyset,\emptyset)$.
Then, while $|A_i|>\frac{2}{3} \cdot |V(G)|$, we apply either \Cref{pro:leightrao} or \Cref{pro:leekpr} to the induced subgraph $G[A_i]$ with the parameter $\cng$.
If we get a concurrent flow of congestion $\cng$, we are done as we are in the case of \Cref{enum:ret:conflow}.
Also, if \Cref{pro:leekpr} returns the graph $\onedot{K_h}$ as an induced minor, we are in the case of \Cref{enum:ret:minor} as it can be turned into an induced minor model of $H$.
If we obtain a separation $(X_i,Y_i,Z_i)$ of $G[A_i]$ of sparsity $\OO(\mu /\cng)$, we construct $(A_{i+1},S_{i+1},B_{i+1})$ as follows.
Let $|X_i| \ge |Z_i|$.
We set $A_{i+1}=X_i$, $S_{i+1} = S_i \cup Y_i$, and $B_{i+1} = B_i \cup Z_i$ and observe that $(A_{i+1}, S_{i+1}, B_{i+1})$ is indeed a separation.

Let $(A_t,S_t,B_t)$ be the separation we obtain once the process stops, i.e., when $|A_t| \le \frac{2}{3} \cdot |V(G)|$.
Let us bound $|S_t|$.
Because the separations have sparsity $\OO(\mu /\cng)$, for every separation $(X_i,Y_i,Z_i)$ encountered in this process it holds that
\begin{align*}
|Y_i| &\le \OO(\mu \cdot |X_i \cup Y_i| \cdot |Y_i \cup Z_i|/\cng) \le |Y_i \cup Z_i| \cdot \OO(\mu \cdot |V(G)|/\cng).
\end{align*}
As the sum of $|Y_i \cup Z_i|$ over all iterations is at most $|V(G)|$, we have that the sum of $|Y_i|$ over all iterations is at most $\OO(\mu \cdot |V(G)|^2/\cng)$.

It remains to argue that $|B_t| \le \frac{2|V(G)|}{3}$.
If $|X_{t-1}| \ge \frac{|V(G)|}{3}$, then this holds as then $|A_t| \ge \frac{|V(G)|}{3}$ and $A_t$ and $B_t$ are disjoint.
Otherwise, $|Z_{t-1}| < \frac{|V(G)|}{3}$, so this follows from $|B_{t-1}| \le \frac{|V(G)|}{3}$.
\end{proof}

\subsection{Induced almost-embedding}
Next we construct from concurrent flow an intermediate object we call an \emph{induced almost-embedding}.
Recall that two paths $P_1$ and $P_2$ are mutually induced if $V(P_1)$ and $V(P_2)$ are disjoint and have no edges between them.
Let $G$ and $H$ be graphs.
We define that a pair $(\vmap, \emap)$ of mappings $\vmap \colon V(H) \rightarrow V(G)$ and $\emap \colon E(H) \rightarrow \paths(G)$ is an induced almost-embedding of $H$ in $G$ if

\begin{enumerate}
\item\label{en:iae:cond1} for every edge $uv \in E(H)$, it holds that $\emap(uv)$ is a $\vmap(u)-\vmap(v)$-path in $G$, and
\item\label{en:iae:cond2} for every pair $(u_1 v_1, u_2 v_2) \in E(H)^2$ of edges of $H$ that are distinct and non-incident in $H$ (i.e., so that $u_1,v_1,u_2,v_2$ are pairwise distinct), it holds that the paths $\emap(u_1 v_1)$ and $\emap(u_2 v_2)$ are mutually induced in $G$.
\end{enumerate}

Note that at this point we do not enforce the mapping $\vmap$ to be injective, nor we impose any constraints on paths that correspond to incident edges in $H$.
Next we show how to construct an induced almost-embedding from a concurrent flow.
Recall that a graph is called \emph{subcubic} if the degree of every vertex is at most three.

\begin{lemma}
\label{lem:fromcftoiae}
There is a randomized polynomial-time algorithm that given a graph $G$, a subcubic graph $H$, and a concurrent flow $\cflow \colon \paths(G) \rightarrow \mathbb{R}_{\ge 0}$ of congestion $\gamma \le |V(G)|^2/(15 \cdot \sqrt{|E(H)|} \cdot \sqrt{|E(G)|})$ in $G$, outputs an induced almost-embedding of $H$ in $G$.
\end{lemma}
\begin{proof}
We will aim to use the constructive version of Lov{\'a}sz local lemma of Moser and Tardos~\cite{DBLP:journals/jacm/MoserT10} (\Cref{pro:lllconstr}).
To faciliate this, let us first construct a polynomial-time mapping $f \colon V(G)^2 \times \mathbb{R}_{[0,1]} \rightarrow \paths(G)$ that is given an ordered pair of vertices $(a,b) \in V(G)^2$ and a real number $x \in \mathbb{R}_{[0,1]}$, and returns an $a-b$-path $P \in \paths_{(a,b)}(G)$, so that when $a,b$ are fixed and $x$ is sampled between $0$ and $1$ uniformly at random, the probability of a path $P \in \paths_{(a,b)}(G)$ is $\Pr[f(a,b,x) = P] = \cflow(P)$.
Observe that this can be constructed from the definition of concurrent flow and the fact that only polynomial number of paths have non-zero $\cflow(P)$.

Then, the set of mutually independent random variables consists of vertex mappings $\vmap(v)$ for each $v \in V(H)$ and real numbers $x_{uv}$ for each $uv \in E(H)$.
In particular, we set $\vmap(v)$ to take values from $V(G)$ uniformly at random, and $x_{uv}$ to take values between $0$ and $1$ uniformly at random.

We define the edge mapping associated with the assignment of the random variables to be $\emap(uv) = f(\vmap(u),\vmap(v),x_{uv})$ for all $uv \in E(H)$ (here we order $u$ and $v$ so that $u<v$ in some arbitrary total order of $V(H)$).
Note that the pair $(\vmap,\emap)$ satisfies \Cref{en:iae:cond1} of the definition of induced almost-embedding.
Now, the events will correspond to \Cref{en:iae:cond2} of the definition.
In particular, when $u_1 v_1, u_2 v_2 \in E(H)$ are two distinct non-incident edges of $H$, we say that $u_1 v_1$ and $u_2 v_2$ \emph{collide} if the paths $\emap(u_1 v_1)$ and $\emap(u_2 v_2)$ are not mutually induced.
Now, for each pair of distinct non-incident edges, we have an event that they collide.
We observe that if none of the events occur, then $(\vmap,\emap)$ satisfies also \Cref{en:iae:cond2} of the definition, and therefore is a induced almost-embedding of $H$ in $G$.

Let $u_1 v_1, u_2 v_2 \in E(H)$ be two distinct non-incident edges of $H$, and let us now bound the probability that $u_1 v_1$ and $u_2 v_2$ collide.
Let $(w,z) \in V(G)^2$ be an ordered pair of vertices of $G$ so that $w = z$ or $wz \in E(G)$, and let us say that $u_1 v_1$ and $u_2 v_2$ collide at $(w,z)$ if $w \in V(\emap(u_1 v_1))$ and $z \in V(\emap(u_2 v_2))$.
We observe that the probability that $w \in V(\emap(u_1 v_1))$ is

\begin{align*}
\Pr[w \in V(\emap(u_1 v_1))] &= \frac{1}{|V(G)|^2} \sum_{(a,b) \in V(G)^2} \sum_{P \in \paths_{(a,b)}(G) \text{ and } w \in V(P)} \cflow(P)\\
&\le \frac{\cng}{|V(G)|^2} \le \frac{1}{15 \cdot \sqrt{|E(H)|} \cdot \sqrt{|E(G)|}}.
\end{align*}
The same bound holds for the probability $\Pr[z \in V(\emap(u_2 v_2))]$.
Now, the crucial observation is that because the vertices $u_1,v_1,u_2,v_2$ are distinct, the events $w \in V(\emap(u_1 v_1))$ and $z \in V(\emap(u_2 v_2))$ are independent of each other.
In particular, $w \in V(\emap(u_1 v_1))$ depends on $\vmap(u_1)$, $\vmap(v_1)$, and $x_{u_1 v_1}$, and $z \in V(\emap(u_2 v_2))$ depends on $\vmap(u_2)$, $\vmap(v_2)$, and $x_{u_2 v_2}$.
Therefore, the probability that $u_1 v_1$ and $u_2 v_2$ collide at $(w,z)$ is

\begin{align*}
\Pr[w \in V(\emap(u_1 v_1)) \text{ and } z \in V(\emap(u_2 v_2))] &= \Pr[w \in V(\emap(u_1 v_1))] \cdot \Pr[z \in V(\emap(u_2 v_2))]\\
&\le \frac{1}{225 \cdot |E(H)| \cdot |E(G)|}.
\end{align*}

We have that $u_1 v_1$ and $u_2 v_2$ collide if and only if they collide at some pair $(w,z)$ with $w=z$ or $wz \in E(G)$, and therefore by union bound, the probability that $u_1 v_1$ and $u_2 v_2$ collide is

\begin{align*}
\Pr[u_1 v_1 \text{ and } u_2 v_2 \text{ collide}] &\le \frac{|V(G)|+2 |E(G)|}{225 \cdot |E(H)| \cdot |E(G)|} \le \frac{1}{56 \cdot |E(H)|}.
\end{align*}

Here we used that $G$ is connected and has at least one edge to say that $|V(G)| \le 2 \cdot |E(G)|$.

Let us then bound the number of events that are adjacent to the event that $u_1 v_1$ and $u_2 v_2$ collide.
We observe that an event is adjacent to that if and only if it involves at least one of the vertices $\{u_1, v_1, u_2, v_2\}$.
Because $H$ is subcubic, the number of edges that involve at least one such vertex is at most $12$, and therefore as any event adjacent to the event of $u_1 v_1$ and $u_2 v_2$ colliding must involve at least one such edge, the number of events adjacent to $u_1 v_1$ and $u_2 v_2$ colliding is at most $12 \cdot |E(H)|$.
The probability of an event is $1/(56 \cdot |E(H)|) \le \frac{1}{4 \cdot 12 \cdot |E(H)|}$, so the algorithm of \Cref{pro:lllconstr} indeed finds an assignment so that none of the events occur.
\end{proof}

\subsection{Induced minor model}
Finally, we construct an induced minor model from an induced almost-embedding.
Recall that $\onedot{H}$ denotes the graph $H$ with each edge subdivided once, and $\dotdot{H}$ the graph $H$ with each edge subdivided twice.

\begin{lemma}
\label{lem:fromiaetoimm}
There is a polynomial-time algorithm that given an induced almost-embedding $(\vmap, \emap)$ of a graph $\dotdot{H}$ that has no isolated vertices into a graph $G$, returns an induced minor model of $\onedot{H}$ in $G$.
\end{lemma}
\begin{proof}
First, consider a non-subdivision vertex $v \in V(H) \cap V(\onedot{H}) \cap V(\dotdot{H})$, i.e., a vertex of $\onedot{H}$ that corresponds to some vertex of $H$.
We assign as the branch set $\branchset(v)$ of $v$ the union of the paths corresponding to edges incident to $v$ in $\dotdot{H}$.
More formally, we let
\begin{align*}
\branchset(v) = \bigcup_{u \in N_{\dotdot{H}}(v)} V(\emap(uv)).
\end{align*}
The set $\branchset(v)$ is non-empty because $\dotdot{H}$ has no isolated vertices.
The induced subgraph $G[\branchset(v)]$ is connected because each $G[V(\emap(uv))]$ is connected and contains $\vmap(v)$.
Moreover, for distinct $u,v \in V(H) \cap V(\onedot{H}) \cap V(\dotdot{H})$, the induced subgraphs $G[\branchset(u)]$ and $G[\branchset(v)]$ are disjoint and non-adjacent because the edges incident to $u$ in $\dotdot{H}$ are non-incident to the edges incident to $v$ in $\dotdot{H}$, and therefore the paths forming $\branchset(v)$ are mutually induced with the paths forming $\branchset(u)$.

Then, consider a subdivision vertex $w \in V(\onedot{H})$ of $\onedot{H}$ that corresponds to an edge $uv \in E(H)$.
Let $u,x,y,v$ be the 4-vertex path in $V(\dotdot{H})$ corresponding to the edge $uv$.
Now, the graph $G[V(\emap(ux)) \cup V(\emap(xy)) \cup V(\emap(yv))]$ is connected and contains both $\vmap(u)$ and $\vmap(v)$.
Let $P$ be an arbitrary $\vmap(u)-\vmap(v)$-path that is contained in $G[V(\emap(ux)) \cup V(\emap(xy)) \cup V(\emap(yv))]$.
Because $\emap(ux)$ and $\emap(yv)$ are mutually induced and $P$ starts in $V(\emap(ux))$ and ends in $V(\emap(yv))$, the path $P$ must contain some subpath $P'$ so that $V(P') \subseteq V(\emap(xy))$ and $V(P')$ is adjacent to both $V(\emap(ux))$ and $V(\emap(yv))$.

We set the branch set of $w$ as $\branchset(w) = V(P')$.
Clearly, $G[\branchset(w)]$ is connected and $\branchset(w)$ is adjacent to both $\branchset(u)$ and $\branchset(v)$.
Because $\branchset(w) \subseteq V(\emap(xy))$ and the edge $xy$ is incident only to edges $ux$ and $yv$, the only branch sets that $\branchset(w)$ is adjacent to are $\branchset(u)$ and $\branchset(v)$.
\end{proof}

Next we put \Cref{lem:fromcftoiae,lem:fromiaetoimm} together as the following lemma.
  
\begin{lemma}
\label{lem:fromcftoimm}
There is a randomized polynomial-time algorithm that given graphs $G$ and $H$, and a concurrent flow $\cflow \colon \paths(G) \rightarrow \mathbb{R}_{\ge 0}$ of congestion $\gamma \le |V(G)|^2/(40 \cdot \sqrt{|V(H)| + |E(H)|} \cdot \sqrt{|E(G)|})$ in $G$, outputs an induced minor model of $H$ in $G$.
\end{lemma}
\begin{proof}
We first construct from $H$ a graph $H'$ with no isolated vertices by adding an adjacent vertex to each isolated vertex.
It holds that $H'$ contains $H$ as an induced minor, $|V(H')| \le 2 |V(H)|$, and $|E(H')| \le |E(H)|+|V(H)|$.
Then, we construct from $H'$ a subcubic graph $H''$ that contains $H'$ as an induced minor by replacing vertices of degree more than three by trees.
It holds that $|V(H'')| \le 2 |V(H)|+|E(H)|$ and $|E(H'')| \le 2 |E(H)| + |V(H)|$.

Then, we have that $|E(\dotdot{H''})| \le 3 (2|E(H)| + |V(H)|) \le 6 \cdot (|V(H)|+|E(H)|)$, and therefore
\begin{align*}
\gamma &\le |V(G)|^2/\left(40 \cdot \sqrt{|V(H)| + |E(H)|} \cdot \sqrt{|E(G)|}\right)\\
&\le |V(G)|^2/\left(\frac{40}{\sqrt{6}} \cdot \sqrt{|E(\dotdot{H''})|} \cdot \sqrt{|E(G)|}\right)\\
&\le |V(G)|^2/\left(15 \cdot \sqrt{|E(\dotdot{H''})|} \cdot \sqrt{|E(G)|}\right).
\end{align*}
This implies that the algorithm of \Cref{lem:fromcftoiae} can be used to find an induced almost-embedding of $\dotdot{H''}$ in $G$.
Then, \Cref{lem:fromiaetoimm} can be used to turn it into an induced minor model of $\onedot{H''}$, which can be turned into an induced minor model of $H''$, which can be turned into an induced minor model of $H$.
\end{proof}

We are now ready to combine \Cref{lem:obtconflow,lem:fromcftoimm} into \Cref{thm:septheorem}.

\septheorem*
\begin{proof}
We first apply \Cref{lem:obtconflow} with $G$, $H$, and
\begin{align*}
\cng = \frac{|V(G)|^2}{120 \cdot \sqrt{|V(H)| + |E(H)|} \cdot \sqrt{|E(G)|}}.
\end{align*}
If it returns a balanced separation of $G$ or an induced minor model of $H$ in $G$, we are done, so it remains to consider the case when it returns an induced subgraph $G'$ of $G$ with $|V(G')| \ge \frac{2 |V(G)|}{3}$ together with a concurrent flow of congestion $\cng$ in $G'$.

Now, because $|V(G')| \ge \frac{2 |V(G)|}{3}$ and $|E(G')| \le |E(G)|$, we have that
\begin{align*}
\cng \le \frac{|V(G')|^2}{40 \cdot \sqrt{|V(H)| + |E(H)|} \cdot \sqrt{|E(G')|}},
\end{align*}
and therefore we can use \Cref{lem:fromcftoimm} to construct an induced minor model of $H$ in $G'$, which is also an induced minor model of $H$ in $G$.
\end{proof}

\section{Subexponential algorithms}
\label{sec:algos}
This section is devoted to proving \Cref{cor:algos,cor:indminregoc}.
We first give a generic meta-theorem for obtaining subexponential time algorithms by combining separator theorems and algorithms parameterized by treewidth, and then apply it to obtain \Cref{cor:algos,cor:indminregoc} from \Cref{thm:septheorem}.
Although there exists a generic meta-theorem~\cite{DBLP:journals/algorithmica/NovotnaOPRLW21} for obtaining subexponential time algorithms from separator theorems similar to our \Cref{thm:septheorem}, it appears to only yield $2^{\OO(n^{3/4} \polylog\ n)}$ time algorithms from $\OO(\sqrt{m})$ separator theorems.
We improve this to $2^{\OO(n^{2/3} \log n)}$.

\subsection{Improved meta-theorem for subexponential algorithms from separator theorems}
We start with a general lemma about using branching to cover degenerate induced subgraphs.
This ``degeneracy-branching'' idea is inspired by the algorithm of~\cite{DBLP:conf/stoc/GartlandLPPR21}.

\begin{lemma}
\label{lem:degbranch}
There is an algorithm that given an $n$-vertex graph $G$ and two integers $\Delta \ge \delta \ge 0$, in time $n^{n (\delta+1)^2/(\Delta-\delta+1)} \cdot n^{\OO(1)}$ outputs a family of at most $n^{(\delta+1)^2 n/(\Delta-\delta+1)}$ vertex subsets $X_1, \ldots, X_t \subseteq V(G)$ so that
\begin{enumerate}
\item\label{lem:degbranch:pro1} for every $Y \subseteq V(G)$ so that $G[Y]$ has degeneracy at most $\delta$, there exists $X_i \supseteq Y$, and
\item\label{lem:degbranch:pro2} for every $X_i$, there exists a set $Z_i \subseteq X_i$ of size at most $(\delta+1) n/(\Delta-\delta+1)$ so that $G[X_i \setminus Z_i]$ has maximum degree at most $\Delta$.
\end{enumerate}
\end{lemma}
\begin{proof}
We will describe a branching algorithm that maintains a partition of the vertex set $V(G)$ into $I$ (in), $O$ (out), and $U$ (undecided), and a ``potential function'' $\phi : V(G) \rightarrow [-1, \delta]$ on vertices.
The sets $X_i$ outputted by the algorithm will correspond to the sets $I \cup U$ in the leaves of this branching.
In the beginning, all vertices of $G$ are in the set $U$, and $\phi(v) = \delta$ for all vertices $v$.
The intuition will be that the function $\phi(v)$ counts how many more left neighbors the vertex $v$ is allowed to have in the set $I$.

Next we describe the branching.
First, if the graph $G[U]$ has maximum degree at most $\Delta$, then we output $I \cup U$ as one of the sets $X_i$.
Otherwise, there exists $v \in U$ so that $v$ has degree more than $\Delta$ in $G[U]$.
We select an arbitrary such $v$, and branch on it in the following sense.
In the first branch, we simply move $v$ from $U$ to $O$.
In the second branch, we further branch on all possibilities of a subset $N_l \subseteq N(v) \cap U$ of size at most $|N_l| \le \delta$, and in each of these branches move $v$ from $U$ to $I$, and then decrease $\phi(u)$ by one for all $u \in (N(v) \cap U) \setminus N_l$.
After that, we move all vertices $u \in U$ with $\phi(u) = -1$ to the set $O$.

Let us then analyze the correctness of this branching, i.e., show that \Cref{lem:degbranch:pro1} holds.
Let us fix some $Y \subseteq V(G)$ so that $G[Y]$ has degeneracy at most $\delta$, and show that some $X_i$ with $Y \subseteq X_i$ is outputted by the algorithm.
This follows from the following claim.
\begin{claim}
Consider some fixed $\delta$-degeneracy ordering $\eta$ of $G[Y]$.
If $Y \cap O = \emptyset$ and for all $v \in U \cap Y$ it holds that $\phi(v)$ is at least the number of left neighbors of $v$ with respect to $\eta$ in $U$, then this holds also for at least one child branch. 
\end{claim}
\begin{claimproof}
Let $v \in U$ be the vertex we are branching on.
If $v \notin Y$, then this holds for the child branch where we move $v$ from $U$ to $O$.
Otherwise, consider the branch where $N_l$ consists of the left neighbors of $v$ with respect to $\eta$ in $U$.
Such a branch exists, because the number of left neighbors of $v$ is at most $\delta$.
For all other neighbors $u \in (N(v) \cap U) \setminus N_l$, it holds that if $u \in Y$, then $v$ is a left neighbor of $u$ with respect to $\eta$, so we can decrease $\phi(u)$ for them without losing the invariant.
From this invariant, it also follows that vertices $u \in U$ whose $\phi(u)$ decreased from $0$ to $-1$ cannot be in $Y$, and therefore can safely be moved to $O$.
\end{claimproof}

Then, we prove that \Cref{lem:degbranch:pro2} holds.
For this, let us consider the function $\phi^+ : V(G) \rightarrow [0,\delta+1]$ defined by $\phi^+(v) = \phi(v)+1$, and let $\phi^+(U) = \sum_{v \in U} \phi^+(v)$.
We observe that at the beginning, the value of $\phi^+(U)$ is $(\delta+1) \cdot n$, and in all branches of the second type, the value $\phi^+(U)$ decreases by at least $\Delta-\delta+1$.
Therefore, in any root-leaf path of the branching tree, there can be at most $(\delta+1)n/(\Delta-\delta+1)$ branches of the second type.
This implies that in each leaf, the set $I$ has size at most $(\delta+1)n/(\Delta-\delta+1)$, so when we take $X_i = I \cup U$ and $Z_i = I$, we get that $G[X_i \setminus Z_i] = G[U]$ has maximum degree $\Delta$.

Then we prove the time complexity and the bound on the number of sets $X_i$.
Any root-leaf path of the branching tree can be described by a string of length at most $n$, where at most $(\delta+1)n/(\Delta-\delta+1)$ points are marked as the branches of second type, and each of them is marked with one choice out of at most $n^{\delta}$.
This yields that the number of root-leaf paths is at most 
\begin{align*}
n^{(\delta+1)n/(\Delta-\delta+1)} \cdot n^{\delta \cdot (\delta+1)n/(\Delta-\delta+1)} =  n^{(\delta+1)^2 n/(\Delta-\delta+1)}.
\end{align*}
This directly bounds the number of sets $X_i$, and bounds the time complexity by observing that each node of the branching tree can be handled in polynomial-time.
\end{proof}

We say that a graph class $\mathcal{C}$ is $\OO(\sqrt{m})$-separable if it is hereditary and and there exists a constant $c$ so that every graph $G \in \mathcal{C}$ has a balanced separator with at most $c \cdot \sqrt{|E(G)|}$ vertices.
Dvor{\'{a}}k and Norin showed that  if every subgraph of $G$ has a balanced separator of size $a$ then the treewidth of $G$ is at most $15a$. From this we immediately obtain an upper bound on the treewidth of  $\OO(\sqrt{m})$-separable classes (We remark that Proposition~\ref{pro:twbound} also admits a fairly short elementary proof based on recursively applying separators of size $\OO(\sqrt{n\Delta})$).


\begin{proposition}
\label{pro:twbound}
Let $\mathcal{C}$ be a $\OO(\sqrt{m})$-separable graph class, and $G$ a graph in $\mathcal{C}$ with $n$ vertices and maximum degree $\Delta$.
The graph $G$ has treewidth at most $\OO(\sqrt{n \Delta})$.
\end{proposition}
\begin{proof}
By definition, all induced subgraphs of $G$ have balanced separators of size at most $\OO(\sqrt{m}) \le \OO(\sqrt{n \Delta})$.
By~\cite{DBLP:journals/jct/DvorakN19} this implies that the treewidth is at most $\OO(\sqrt{n \Delta})$.
\end{proof}

We say that a problem $\Pi$ is an induced subgraph problem if the input is a graph $G$, and the task is to decide whether there exists a set of vertices $X \subseteq V(G)$ so that $G[X]$ satisfies some graph property (i.e., a property that depends only the induced subgraph $G[X]$ and is the same for isomorphic graphs).
We say a induced subgraph problem $\Pi$ is degenerate if there exists a constant $\delta$ so that all graphs $G[X]$ that satisfy the property have degeneracy at most $\delta$.
By applying \Cref{lem:degbranch}, we obtain the following meta-theorem about subexponential time algorithms.

\begin{theorem}
\label{the:metathe}
Let $\mathcal{C}$ be a $\OO(\sqrt{m})$-separable graph class and $\Pi$ a degenerate induced subgraph problem that can be solved in time $2^{\OO(\tw \log \tw)} n^{\OO(1)}$ on $n$-vertex graphs of treewidth $\tw$.
Then, $\Pi$ can be solved in time $2^{\OO(n^{2/3} \log n)}$ on $n$-vertex input graphs from the class $\mathcal{C}$.
\end{theorem}
\begin{proof}
We describe the $2^{\OO(n^{2/3} \log n)}$ time algorithm.
Let $\delta$ be the degeneracy bound in the problem $\Pi$, and let us first apply \Cref{lem:degbranch} with this $\delta$ and $\Delta = n^{1/3}$.
Assuming $\delta$ is constant and $n$ is superconstant, the running time and the number of sets $X_i$ outputted is at most
\begin{align*}
n^{n (\delta+1)^2/(\Delta-\delta+1)} \cdot n^{\OO(1)} \le 2^{\OO(n^{2/3} \log n)}.
\end{align*}
Now, for each $X_i$ there exists $Z_i \subseteq X_i$ of size at most $\OO(n/\Delta) = \OO(n^{2/3})$ so that $G[X_i \setminus Z_i]$ has maximum degree $\Delta$.
By \Cref{pro:twbound}, $G[X_i \setminus Z_i]$ has treewidth $\OO(\sqrt{n \Delta}) = \OO(n^{2/3})$, and therefore $G[X_i]$ also has treewidth $\OO(n^{2/3})$.
Now, we can use the algorithm parameterized by treewidth to solve the problem on each $G[X_i]$ in time $2^{\OO(n^{2/3} \log n)}$.
The total running time is 
\begin{align*}
2^{\OO(n^{2/3} \log n)} \cdot 2^{\OO(n^{2/3} \log n)} \le 2^{\OO(n^{2/3} \log n)}.
\end{align*}
The correctness follows from \Cref{lem:degbranch:pro2} of \Cref{lem:degbranch}.
\end{proof}

We remark that for problems admitting $2^{\OO(\tw)} n^{\OO(1)}$ time algorithms, the running time obtained from \Cref{the:metathe} could be improved slightly, but for simplicity we prefer to state only one variant of the theorem.

\subsection{Applications}
We then use \Cref{the:metathe} to prove \Cref{cor:algos}.

For a collection $\mathcal{F}$ of graphs, the $\mathcal{F}$-\textsc{Minor-Deletion} problem asks for a minimum set of vertices to delete so that the resulting graph does not contain any graph from $\mathcal{F}$ as a minor.
Equivalently, it asks for a maximum induced subgraph that is $\mathcal{F}$-minor-free.
As $\mathcal{F}$-minor-free graphs have bounded degeneracy~\cite{DBLP:journals/combinatorica/Kostochka84,thomason_1984} and $H$-induced-minor-free graphs are $\OO(\sqrt{m})$-separable by \Cref{thm:septheorem}, we immediately obtain the following corollary of \Cref{the:metathe} by using a $2^{\OO_{\mathcal{F}}(\tw \log \tw)} n^{\OO(1)}$ algorithm of~\cite{DBLP:conf/soda/BasteST20}.

\begin{corollary}
Let $\mathcal{F}$ be a fixed collection of connected graphs.
The $\mathcal{F}$-\textsc{Minor-Deletion} problem admits a $2^{\OO_{\mathcal{F},H}(n^{2/3} \log n)}$ time algorithm on $H$-induced-minor-free graphs.
\end{corollary}

We remark that $\mathcal{F}$-\textsc{Minor-Deletion} encompasses \textsc{Maximum Independent Set}, \textsc{Minimum Feedback Vertex Set}, and \textsc{Planarization}.

The \textsc{Maximum Induced Matching} problem admits a $2^{\OO(\tw)} n$ time algorithm~\cite{DBLP:journals/dam/MoserS09}, so we also obtain the following.
\begin{corollary}
\textsc{Maximum Induced Matching} can be found in time $2^{\OO_{H}(n^{2/3} \log n)}$ on $H$-induced-minor-free graphs.
\end{corollary}

To complete the proof of \Cref{cor:algos}, we observe that the algorithm of Bonnet and Rzazewski~\cite{DBLP:journals/algorithmica/BonnetR19} for $3$-\textsc{Coloring} on string graphs is in fact an algorithm for $\OO(\sqrt{m})$-separable graphs; it does not use any other properties of string graphs than the separator theorem.

\begin{proposition}[\cite{DBLP:journals/algorithmica/BonnetR19}]
There is a $2^{\OO(n^{2/3} \log n)}$ time algorithm for $3$-\textsc{Coloring} on $\OO(\sqrt{m})$-separable graphs.
\end{proposition}

Therefore, we obtain the last piece of \Cref{cor:algos}.

\begin{corollary}
There is a $2^{\OO_{H}(n^{2/3} \log n)}$ time algorithm for $3$-\textsc{Coloring} on $H$-induced-minor-free graphs.
\end{corollary}

Then, let us turn to the proof of \Cref{cor:indminregoc}, which states that for graph $H$ where every edge is incident to a vertex of degree at most $2$, there is a $2^{\OO_{H}(n^{2/3} \log n)}$ time algorithm for testing if a given $n$-vertex graph contains $H$ as a minor.

First, we recall that there is an algorithm parameterized by treewidth for testing if a given graph contains $H$ as an induced minor~\cite{DBLP:journals/tcs/AdlerDFST11}.

\begin{proposition}[\cite{DBLP:journals/tcs/AdlerDFST11}]
\label{pro:twindminortesting}
There is a $2^{\OO_H(\tw \log \tw)} n^{\OO(1)}$ time algorithm for testing if a given $n$-vertex graph of treewidth $\tw$ contains $H$ as an induced minor.
\end{proposition}

It remains to give the following structural result about minimal graphs that contain induced minors $H$ where every edge is incident to a vertex of degree at most $2$.

\begin{lemma}
\label{lem:sparsemodels}
Let $H$ be a graph where every edge is incident to a vertex of degree at most $2$ and $G$ any graph.
If $G$ contains $H$ as an induced minor, but no induced subgraph of $G$ contains $H$ as an induced minor, then the degeneracy of $G$ is at most $3 \cdot |V(H)|$.
\end{lemma}
\begin{proof}
Let us consider an induced minor model $\{X_v \mid v \in V(H)\}$ of $H$ in $G$.
First, we show that each $G[X_v]$ has maximum degree at most $|V(H)|$.

Suppose $u \in X_v$ has degree more than $|V(H)|$ in $G[X_v]$.
For each neighbor $w$ of $v$ in $H$ we can mark a vertex in $X_v$ that is adjacent to $X_w$.
Now, if there would be a strict subset $X'_v \subset X_v$ that would contain these marked vertices and $G[X'_v]$ would be connected, then we would contradict the minimality of $G$.
However, when $u \in X_v$ has degree more than $|V(H)|$, such a subset can be created by taking the union of shortest paths from $u$ in $G[X_v]$ to each of the marked vertices.
Therefore, no such $u$ exists.

Then, we show that if $u,v \in V(H)$ are adjacent vertices in $H$ and $v$ has degree at most $2$, then exactly one vertex in $X_v$ is adjacent to $X_u$.
First, if $v$ has degree $1$, then the minimality of $G$ implies that $|X_v| = 1$, and therefore this holds.
Then, assume $v$ has degree $2$, and in particular, is adjacent to $w \in V(H)$ in addition to $u$.

First, if there would be a vertex in $X_v$ that would be adjacent to both $X_u$ and $X_w$, then $|X_v| = 1$ by the minimality of $G$.
Then, if there would be two vertices $a,b \in X_v$ that are both adjacent to $X_u$, then either of them could be removed by considering a vertex $c \in X_v$ adjacent to $X_w$, and considering a shortest $a-c$-path $P_a$ in $G[X_v]$ and a shortest $b-c$-path $P_b$ in $G[X_v]$, and observing that we can set either $X_v = V(P_a)$ or $X_v = V(P_b)$, and either $P_a$ does not contain $b$ or $P_b$ does not contain $a$.
Therefore, by the minimality of $G$, only a single vertex in $X_v$ can be adjacent to $X_u$.

Now, after removing from all branch sets corresponding to vertices of degree at most $2$ in $H$ the (at most two) vertices that are adjacent to other branch sets, the resulting graph has maximum degree $|V(H)|$.
This implies that the degeneracy of $G$ is at most $3 \cdot |V(H)|$.
\end{proof}

We remark that with a bit of work, even much stronger sparsity properties could be proven for such $G$, but for applying \Cref{the:metathe} it is sufficient to bound the degeneracy by a function of $H$.
By combining \Cref{pro:twindminortesting}, \Cref{lem:sparsemodels}, and \Cref{the:metathe,thm:septheorem}, we finally obtain the following.

\testingcor*

Note that even though a direct application of \Cref{the:metathe} implies only a decision version of \Cref{cor:indminregoc}, this can be turned into an algorithm that actually returns the induced minor model by standard self reduction. 

\section{Hardness of Induced Minor Testing}
\label{sec:hardness}
In this section we prove \Cref{thm:hardnesstheorem}, in particular, that there exists a fixed tree $T$ so that assuming ETH there is no $2^{o(n / \log^3 n)}$ time algorithm for testing if a given $n$-vertex graph contains $T$ as an induced minor.
Our proof will be in three parts.
First, in \Cref{subsec:binshift} we introduce graphs that we call \emph{binary shift graphs}, prove that they admit a certain type of edge partition, and show that $3$-domain binary CSP admits no $2^{o(n / \log^3 n)}$ time algorithm on them.
Then, in \Cref{subsec:mulcolidp} we reduce $3$-domain binary CSP on binary shift graphs to a certain variant of induced disjoint paths problem with a constant number of paths.
Finally, in \Cref{subsec:hardnessfinalindminor} we reduce this variant of induced disjoint paths to $T$-induced-minor testing for a fixed tree $T$.

\subsection{Binary shift graphs}
\label{subsec:binshift}
We define a class of graphs that we call \emph{binary shift graphs}.
These graphs will have two desirable properties: they will be almost expanders in that they admit a concurrent flow of congestion $\OO(n \log n)$, and their edges can be partitioned into a constant number of sets that satisfy a certain sparsity condition.

For integer $b \ge 1$, the binary shift graph $\BS_b$ has $2^b$ vertices $v_0, \ldots, v_{2^b-1}$.
The binary shift graph $\BS_b$ has an edge between vertices $v_x$ and $v_y$ if $x \neq y$ and either
\begin{itemize}
\item $x \equiv 2y \mod{2^b}$,
\item $x \equiv 2y+1 \mod{2^b}$,
\item $y \equiv 2x \mod{2^b}$, or
\item $y \equiv 2x+1 \mod{2^b}$.
\end{itemize}
In particular, $v_x$ is adjacent to $v_y$ if the length-$b$ binary representation of $y$ can be obtained from the length-$b$ binary representation of $x$ by ``shifting'' it by one digit to left or right and setting the new digit to be either $0$ or $1$.

Next we give the partition of the edges of the binary shift graph to sets that satisfy a certain pathwidth bounding condition.

\begin{lemma}
\label{lem:binshiftpartition}
Let $b \ge 1$ be an integer and $P$ a path on $2^b$ vertices $v_0, \ldots, v_{2^b-1}$ where $v_i$ is adjacent to $v_j$ if $|i-j|=1$.
The edges $E(\BS_b)$ of the binary shift graph $\BS_b$ can be partitioned into four sets $E_1, E_2, E_3, E_4$ so that for each $i \in [4]$, the graph $P \cup E_i$ has pathwidth at most $16$.
Moreover, there is a polynomial-time algorithm that given $\BS_b$ outputs this partition.
\end{lemma}
\begin{proof}
We first observe that if the edges of type $x \equiv 2y \mod{2^b}$ can be partitioned into $4$ such sets with pathwidth at most $7$, then all edges can be partitioned into $4$ such sets with pathwidth at most $16$ by inserting $v_{i+1}$ to every bag that contains $v_i$ for all $i \in [0,2^b-2]$, and inserting $v_{0}$ to all bags.

Then, we observe that because of the symmetry of the operation $2y \mod{2^b}$ on the interval $y \in [1,2^{b-1}-1]$ with the interval $y \in [2^{b-1}+1, 2^b-1]$, to get the desired partition it suffices to partition the edges of type $x \equiv 2y \mod{2^b}$ with $y \in [1,2^{b-1}-1]$ into $2$ such sets with pathwidth at most $5$ (then we insert $v_0$ and $v_{2^{b-1}}$ into all bags to arrive at pathwidth $7$).
Next we give this partition.

We first partition the integers $[1, 2^b-1]$ into $b$ intervals
\[[1,1], [2,3], [4,7], \ldots, [2^i, 2^{i+1}-1], \ldots, [2^{b-1}, 2^b-1].\]
In particular, for $i \in [0,b-1]$, the $i$-th interval is $[2^i, 2^{i+1}-1]$.
Now, we have that any edge of type $x \equiv 2y \mod{2^b}$ with $y \in [1, 2^{b-1}-1]$ is between the $i$-th interval and $i+1$-th interval for some $i \in [0,b-2]$.
Then, we say that the $i$-th interval is even if $i$ is even and odd if $i$ is odd.
We partition the edges of type $x \equiv 2y \mod{2^b}$ with $y \in [1, 2^{b-1}-1]$ into two parts $E_e$ and $E_o$ based on whether $y$ is in an even interval or in an odd interval. 
It remains to argue that $P \cup E_e$ and $P \cup E_o$ have pathwidth at most $5$.

For an induced subgraph of $P \cup E_e$ whose vertices correspond to an union of an even interval $i$ and an odd interval $i+1$, we construct a path decomposition of width $5$ by having a path of bags $B_{2^{i}}, \ldots, B_{2^{i+1}-1}$ indexed by the integers in the even interval, and having $B_y$ for $y \in [2^i, 2^{i+1}-1]$ contain $v_y$, $v_{y+1}$, $v_{2y}$, $v_{2y+1}$, and $v_{2y+2}$ if $2y+2 < 2^{i+2}$, and additionally every bag to contain $v_{2^{i+1}}$.
It can be verified that this indeed is a path decomposition of such an induced subgraph of $G$ of width $5$.
Now, as the first bag $B_{2^i}$ contains $v_{2^i}$ and the last bag $B_{2^{i+1}-1}$ contains $v_{2^{i+2}-1}$, these path decompositions of the induced subgraphs can be chained together, in particular, observing that $E_e$ does not contain any edges from the union of the intervals $i$ and $i+1$ to outside of these intervals.
This results in a path decomposition of $P \cup E_e$ of width $5$.

The argument for showing that the pathwidth of $P \cup E_o$ is at most $5$ is similar.
It can be observed that all of our arguments are constructive, so the partition into $E_1, E_2, E_3, E_4$ can be constructed in polynomial-time when given $\BS_b$ whose vertices are labeled with $v_0, \ldots, v_{2^b-1}$.
\end{proof}

We then prove that an $n$-vertex binary shift graph is an almost expander in the sense that it admits a concurrent flow of congestion $\OO(n \log n)$.

\begin{lemma}
\label{lem:binshiftconflow}
For every $b \ge 1$, the binary shift graph $\BS_b$ admits a concurrent flow of congestion $b \cdot 2^b$.
\end{lemma}
\begin{proof}
For every ordered pair of vertices $(v_x, v_y) \in V(\BS_b)^2$ we define a $v_x-v_y$-path $P_{(x,y)}$ via which we send one unit of flow.
This path $P_{(x,y)}$ is defined by considering the binary string of length $2b$ that results from concatenating the length-$b$ binary representations of $x$ and $y$, and then taking the vertices corresponding to substrings of length $b$.
This results in a connected subgraph with $b$ vertices that contains both $v_x$ and $v_y$, and the path $P_{(x,y)}$ is constructed by taking an an arbitrary $v_x-v_y$-path in that subgraph.

It remains to bound the congestion.
The congestion at vertex $v_z$ is upper bounded by the number of binary strings of length $2b$ so that the length-$b$ binary representation of $z$ is a substring of that string.
This is at most $b \cdot 2^b$ for any $z \in [0,2^b-1]$.
\end{proof}

We then make use of the expansion of binary shift graphs to give an ETH lower bound for CSP on them.
More formally, we define the problem of $3$-domain binary CSP on a graph class $\gclass$ to have as an input a graph $G \in \gclass$, and a binary CSP with domain size $3$ whose variables correspond to the vertices of $G$ and where there can be arbitrary constraints between two variables if the corresponding vertices are connected by an edge in $G$.
The problem is to decide if the CSP is satisfiable.

Now, the idea is to embed an arbitrary $3$-coloring instance with at most $m$ edges and vertices to a binary shift graph with $\OO(m \log^3 m)$ vertices.
For this, we use the following result of Krivelevich and Nenadov~\cite{DBLP:conf/bcc/Krivelevich19}.

\begin{proposition}[\cite{DBLP:conf/bcc/Krivelevich19}]
\label{pro:minorembed}
There is a polynomial-time algorithm that given an $n$-vertex graph $G$, a rational $\alpha > 0$, and a graph $H$ with at most $\alpha^2 n / \log n$ edges and vertices, outputs either a minor model of $H$ in $G$, or a balanced separator of $G$ of size at most $\OO(\alpha n)$.
\end{proposition}

This results in the following hardness result for $3$-domain binary CSP on binary shift graphs with $n$ vertices.

\begin{lemma}
\label{lem:bincsphardness}
Assuming ETH, there is no $2^{o(n/\log^3 n)}$ time algorithm for $3$-domain binary CSP on binary shift graphs with $n$ vertices.
\end{lemma}
\begin{proof}
Recall that assuming ETH, there is no $2^{o(m)}$ algorithm for $3$-coloring on graphs with at most $m$ edges and vertices~\cite{ImpagliazzoPZ01}.
Therefore, to prove the lemma it suffices to reduce $3$-coloring on arbitrary graphs with at most $m$ edges and vertices to $3$-domain binary CSP on binary shift graphs with $\OO(m \log^3 m)$ vertices.

Observe that if a graph $G$ contains a graph $H$ as a minor, then given the minor model of $H$ in $G$, we can reduce $3$-coloring on $H$ to $3$-domain binary CSP on $G$ by setting equality constraints between variables sharing an edge within the same branch set, and inequality constraints between variables sharing an edge between branch sets that correspond to adjacent vertices of $H$.

Therefore, it suffices to give a polynomial-time algorithm that given a graph $H$ with at most $m$ edges and vertices, finds a minor model of $H$ in a binary shift graph with at most $\OO(m \log^3 m)$ vertices.
For this, we use the algorithm of \Cref{pro:minorembed}.
Recall that $\BS_b$ has $2^b$ vertices and by \Cref{lem:binshiftconflow} admits a concurrent flow of congestion $b \cdot 2^b$.
Therefore, any balanced separator of $\BS_b$ has size at least
\begin{align*}
\left(\frac{2^b}{3}\right)^2 / (b \cdot 2^b) \ge 2^b / (9 b).
\end{align*}

Now, when setting $\alpha = 1/(b \cdot c_1)$ for some large enough constant $c_1$, the algorithm of \Cref{pro:minorembed} can output any graph $H$ with at most 
$\alpha^2 2^b / \log 2^b = 2^b / (b^3 c_1^2)$ edges and vertices as a minor of $\BS_b$.
Therefore, when we take $b = \log m + 3 \log \log m + c_2$ for some large enough constant $c_2$, we have that $\BS_b$ has $\OO(m \log^3 m)$ vertices and contains every graph with at most $m$ edges and vertices as a minor.
This minor model can be found in polynomial time by the algorithm of \Cref{pro:minorembed}.
\end{proof}

Let us remark that in \Cref{lem:bincsphardness} we could have used our \Cref{lem:fromcftoimm} instead of \Cref{pro:minorembed} and obtained a $2^{o(n/ \log^2 n)}$ lower bound under randomized reductions, instead of the $2^{o(n/ \log^3 n)}$ lower bound under deterministic reductions.

\subsection{Multicolored induced disjoint paths}
\label{subsec:mulcolidp}
The next step of our reduction is to reduce $3$-domain binary CSP on binary shift graphs to a problem we call \emph{multicolored induced $k$-disjoint paths}, and in particular, to an instance of this problem with a particular structure.
In the multicolored induced $k$-disjoint paths problem, the input consists of a graph $G$, a partition $V_1, \ldots, V_k$ of $V(G)$ into $k$ parts, and $k$ pairs of terminal vertices $(s_1, t_1), \ldots, (s_k, t_k)$, with $s_i, t_i \in V_i$.
The problem is to decide if there exists $k$ paths $P_1, \ldots, P_k$, so that for each $i \in [k]$, $P_i$ is a $s_i-t_i$-path that is contained in $V_i$, and there are no edges between $P_i$ and $P_j$ for every $i \neq j$.

\begin{lemma}
\label{lem:chaininduced}
There exists constants $k$ and $w$, so that assuming ETH, there is no $2^{o(n/\log^3 n)}$ time algorithm for multicolored induced $k$-disjoint paths on graphs partitioned into $V_1, \ldots, V_k$ so that
\begin{itemize}
\item there are no edges between $V_i$ and $V_j$ when $|i-j| > 1$, and
\item for each $i \in [k-1]$, the graph $G[V_i \cup V_{i+1}]$ has pathwidth at most $w$.
\end{itemize} 
\end{lemma}
\begin{proof}
We reduce from $3$-domain binary CSP on binary shift graphs with $n$ vertices, which by \Cref{lem:bincsphardness} does not admit a $2^{o(n/\log^3 n)}$ time algorithm assuming ETH.
Let $G = \BS_b$ be the input binary shift graph with $n = 2^b$ vertices $v_0, \ldots, v_{2^b-1}$ and $E_1,E_2,E_3,E_4$ the partition of $E(G)$ given by \Cref{lem:binshiftpartition}.
Next we describe the reduction.

The number of paths $k$ in the instance we produce will be $k = 5$.
Let us denote by $G'$ the graph of the instance we produce.
For each $i \in [5]$, the induced subgraph $G'[V_i]$ is constructed as follows (see \Cref{fig:gv1} for the construction of $G'[V_1]$).
It has vertices $a^j_i$, $b^j_i$, and $c^j_i$ for all $j \in [0,n-1]$, vertices $u^j_i$ for all $j \in [0,n-2]$, and vertices $s_i$ and $t_i$.
For each $j \in [1,n-1]$, from each $a^j_i$, $b^j_i$, and $c^j_i$ there are edges to $u^{j-1}_i$.
For each $j \in [0,n-2]$, from each $a^j_i$, $b^j_i$, and $c^j_i$ there are edges to $u^j_i$.
From $a^0_i$, $b^0_i$, and $c^0_i$ there are edges to $s_i$, and from $a^{n-1}_i$, $b^{n-1}_i$, and $c^{n-1}_i$ there are edges to $t_i$.

\begin{figure}[htb]
\begin{center}
\includegraphics[width=0.9\textwidth]{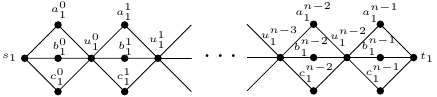}
\caption{The construction of $G'[V_1]$.}
\label{fig:gv1}
\end{center}
\end{figure}

The idea of this construction of $G'[V_i]$ is that any $s_i-t_i$-path in $G'[V_i]$ contains $u^j_i$ for all $j \in [0,n-2]$, and additionally, for each $j \in [0,n-1]$ exactly one of $a^j_i$, $b^j_i$, and $c^j_i$.
The choice of whether the path uses $a^j_i$, $b^j_i$, or $c^j_i$ will correspond to the assignment of the $j$-th variable of the CSP.

Then we describe the edges that go between $V_i$ and $V_{i+1}$.
First, we create copy gadgets to ensure that the assignment corresponding to the path $P_{i+1}$ will be the same as the assignment corresponding to the path $P_i$.
In particular, for each $j \in [0,n-1]$, we add edges from $a^j_i$ to $b^j_{i+1}$ and $c^j_{i+1}$, from $b^j_i$ to $a^j_{i+1}$ and $c^j_{i+1}$, and from $c^j_i$ to $a^j_{i+1}$ and $b^j_{i+1}$.
Then, we add constraint gadgets for constraints whose edges are in $E_i$.
In particular, if there is a constraint between $x$-th and $y$-th variable whose corresponding edge $v_x v_y$ is in $E_i$, then we add edges between $\{a^x_i, b^x_i, c^x_i\}$ and $\{a^y_{i+1}, b^y_{i+1}, c^y_{i+1}\}$ so that any paths $P_i$ and $P_{i+1}$ that would correspond to an assignment of the variables that would violate the constraint would have an edge between them.
This completes the description of the reduction.

We observe that if there is a solution to the CSP, then there is also a solution to the multicolored induced $k$-disjoint paths problem by just taking each $P_i$ to be a path corresponding to the assignment of the variables in the solution.
Also, if there is a solution to the multicolored induced $k$-disjoint paths problem, then by the properties of $G'[V_i]$, each $P_i$ corresponds to an assignment of variables.
By the construction of the copy gadgets the assignments are the same for all $P_i$.
Now, such an assignment must satisfy every constraint, because every constraint corresponds to an edge in some $E_i$, and the assignments corresponding to $P_i$ and $P_{i+1}$ must satisfy this constraint.

By construction, there are no edges between $V_i$ and $V_j$ when $|i-j| > 1$.
It remains to argue that for each $i \in [k-1]$, the graph $G'[V_i \cup V_{i+1}]$ has pathwidth bounded by some constant $w$.
Let $P$ be a path on vertices $v_0, \ldots, v_{n-1}$ so that by \Cref{lem:binshiftpartition}, the graph $P \cup E_i$ has pathwidth at most $16$.
We take a path decomposition of $P \cup E_i$ of width at most $16$ and construct a path decomposition of $G'[V_i \cup V_{i+1}]$ as follows.
For each $j \in [0,n-1]$, we put $a^j_i$, $a^j_{i+1}$, $b^j_i$, $b^j_{i+1}$, $c^j_i$, and $c^j_{i+1}$ to all bags that contain $v_j$.
Then for each $j \in [0,n-2]$, we put $u^j_i$ and $u^j_{i+1}$ to all bags that contain $v_{j}$ or $v_{j+1}$.
Finally, we put $s_i$ and $s_{i+1}$ to all bags that contain $v_0$, and $t_i$ and $t_{i+1}$ to all bags that contain $v_{n-1}$.
We observe that this indeed results in a path decomposition of $G'[V_i \cup V_{i+1}]$ of width at most $(16+1) \cdot (6 + 4)-1 = 169$.
\end{proof}

Then, we reduce the problem of \Cref{lem:chaininduced} to a problem we call \emph{anchored induced minor testing}.
Let $a$ be a fixed integer and $H$ a fixed graph.
In $a$-anchored $H$-induced minor testing, the input consists of a graph $G$ and a list of $a$ anchors $(v_1,u_1), \ldots, (v_a, u_a)$ that are pairs $(v_i,u_i) \in V(G) \times V(H)$.
The problem is to decide if $G$ contains an induced minor model of $H$, so that for every anchor $(v_i,u_i)$ it holds that the branch set of $u_i \in V(H)$ contains the vertex $v_i$ of $G$.

\begin{lemma}
\label{lem:anchored}
There exists a fixed tree $T$ and constants $a$ and $h$, so that assuming ETH, there is no $2^{o(n/\log^3 n)}$ time algorithm for $a$-anchored $T$-induced minor testing on graphs that exclude a complete binary tree of height $h$ as an induced minor.
\end{lemma}
\begin{proof}
We reduce from the problem of \Cref{lem:chaininduced}.
In particular, let $k$ and $w$ be the constants for which this problem does not have a $2^{o(n/\log^3 n)}$ algorithm assuming ETH, and let $G$ be an $n$-vertex input graph for this problem with a partition $V_1, \ldots, V_k$ of $V(G)$ and terminal pairs $(s_1, t_1), \ldots, (s_k, t_k)$.
Without loss of generality, we will assume that $k$ is even and $k \ge 4$.
(Note that $k$ can be increased by adding a set $V_{k+1}$ with a single vertex $s_{k+1}=t_{k+1}$ that is not connected to any other vertex.)

Next we describe the construction of the tree $T$ (see \Cref{fig:tree} for the construction when $k=6$).
The tree $T$ will have $2k$ vertices, $v_1, \ldots, v_k$ and $u_1, \ldots, u_k$.
Each vertex $u_i$ is adjacent only to $v_i$.
Then, for each $i \in [k-2]$, we add an edge between $v_{i}$ and $v_{i+2}$.
In order to make $T$ a tree instead of forest, we also add an edge between $v_1$ and $v_k$.
This completes the description of $T$.
Observe that the vertices $v_1, \ldots, v_k$ form a path $v_2, v_4, \ldots, v_k, v_1, v_3, \ldots, v_{k-1}$ and the vertices $u_i$ have degree $1$, so $T$ is indeed a tree.

\begin{figure}[htb]
\begin{center}
\includegraphics[width=0.6\textwidth]{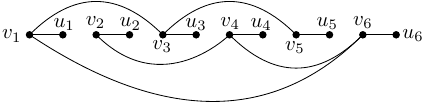}
\caption{The construction of $T$ for $k=6$.}
\label{fig:tree}
\end{center}
\end{figure}

The input graph $G'$ together with $a = 3k$ anchors is constructed as follows.
First, we take $G$ and for every $i \in [k]$ add a vertex $w_i$ that is universal to $V_i$ and non-adjacent to everything else.
For each $i$, we label $w_i$ with an anchor for $u_i$ and both $s_i$ and $t_i$ with an anchor for $v_i$.
Then, for every pair $i,j$ so that there is an edge between $v_i$ and $v_j$ in $T$, we add all possible edges between $V_i$ and $V_j$.
This completes the construction of $G'$.

Next we show that $G'$ contains $T$ as an anchored induced minor if and only if the multicolored induced $k$-disjoint paths instance has a solution.
The if direction follows from setting for each $i \in [k]$ the branch set of the vertex $v_i$ to be the $s_i-t_i$-path, and setting $\{w_i\}$ to be the branch set of $u_i$.
For the only if direction, consider an anchored induced minor model of $T$ in $G'$.
First, we have that the branch set of $v_i$ must be contained in $V_i$ because otherwise $v_i$ would become adjacent to $u_j$ for $j \neq i$.
The branch set of $v_i$ must contain an $s_i-t_i$-path, so we let $P_i$ to be an arbitrary $s_i-t_i$-path contained in the branch set.
As $v_i$ and $v_{i+1}$ are non-adjacent in $T$ for every $i \in [k-1]$, these paths are indeed pairwise-induced in $G$.

It remains to prove that $G'$ excludes some fixed complete binary tree as an induced minor.
For this, we use the following claim.

\begin{claim}
Every induced minor of $G'$ that is a tree has pathwidth at most $3k+w$.
\end{claim}
\begin{claimproof}
For the sake of contradiction, suppose that $H$ is a tree that is an induced minor of $G'$ and has pathwidth more than $3k+w$, and consider the induced minor model of $H$ in $G'$.
For $i \in [k]$, let us call $V_i$ important if at least three branch sets of the induced minor model of $H$ intersect $V_i$, and unimportant otherwise.
Observe that if there is an edge between $v_i$ and $v_j$ in $T$, then at most one of $V_i$ and $V_j$ is important, as otherwise $H$ would contain a cycle.
We delete from the induced minor model of $H$ all branch sets that intersect an unimportant set $V_i$, and all branch sets that intersect a $w_i$ vertex.
This gives an induced minor model of an induced subgraph $H'$ that is obtained by deleting at most $3k$ vertices from $H$, and therefore has pathwidth more than $w$.

Let $H''$ be a connected component of $H'$ with pathwidth more than $w$, and consider the induced minor model of $H''$ given by the induced minor model of $H'$.
Let $i \in [k]$ be the smallest integer so that the induced minor model of $H''$ intersects $V_i$.
We claim that the induced minor model of $H''$ is contained in $V_i \cup V_{i+1}$.
First, because the sets $V_j$ that intersect the model are important, they must correspond to an independent set of $T$.
Therefore, we observe that the only edges of $G'$ that the induced minor model can use are the edges of $G$, and therefore the sets $V_i, \ldots, V_j$ that the model intersects must be consecutive because $H''$ is connected.
However, the model cannot intersect $V_{i+2}$, so the sets are at most $V_i$ and $V_{i+1}$.
Now, $H''$ an induced minor of $G'[V_i \cup V_{i+1}] = G[V_i \cup V_{i+1}]$, and therefore $H''$ has pathwidth at most $w$, which is a contradiction.
\end{claimproof}

Combining with the fact that the complete binary tree of height $6k+2w+2$ has pathwidth more than $3k+w$~\cite{scheffler1989baumweite} and that pathwidth does not increase when taking induced minors, it follows that $G'$ does not contain the complete binary tree of height $h = 6k+2w+2$ as an induced minor.
\end{proof}

\subsection{Induced minors}
\label{subsec:hardnessfinalindminor}
Finally, we reduce the problem of \Cref{lem:anchored} to induced minor testing by replacing the anchors by large enough complete binary trees.

\hardnesstheorem*
\begin{proof}
We reduce from the problem of \Cref{lem:anchored}.
In particular, let $T_0$ be the tree and $a$ and $h$ the constants from the statement of \Cref{lem:anchored}, and let $G_0$ be an $n$-vertex input graph and $(v_1, u_1), \ldots, (v_a, u_a)$ the list of $a$ anchors with $(v_i, u_i) \in V(G_0) \times V(T_0)$.

The idea of our reduction is to attach increasingly large complete binary trees of height more than $h$ to the anchors in both $T_0$ and $G_0$, so that they would guarantee that the branch sets would respect the anchors.
The graph $G_i$ for $i \in [a]$ is constructed by taking the graph $G_{i-1}$ and the complete binary tree $B_{h+2i}$ of height $h+2i$, and identifying the root of $B_{h+2i}$ with the vertex $v_i$ of $G_{i-1}$.
In particular, the vertices of $G_i$ can be partitioned into three sets $L_i$, $\{v_i\}$, and $R_i$, so that $G_i[L_i \cup \{v_i\}] = G_{i-1}$, $G_i[\{v_i\} \cup R_i] = B_{h+2i}$, and there are no edges between $L_i$ and $R_i$.
The tree $T_i$ for $i \in [a]$ is constructed in an analogous manner, by taking the tree $T_{i-1}$ and $B_{h+2i}$, and identifying the root of $B_{h+2i}$ with the vertex $u_i$ of $T_{i-1}$.
Again, the vertices of $T_i$ can be partitioned into $L^T_i$, $\{u_i\}$, and $R^T_i$, so that $T_i[L^T_i \cup \{u_i\}] = T_{i-1}$, $T_i[\{u_i\} \cup R^T_i] = B_{h+2i}$, and there are no edges between $L^T_i$ and $R^T_i$.

The graphs $G_a$ and $T_a$ are our final input graphs.
We observe that when $a$ and $h$ are constants, the graphs $G_a$ and $T_a$ are larger than $G_0$ and $T_0$ by only a constant number of vertices.
We observe that if $G_0$ contains $T_0$ as an anchored induced minor, then $G_a$ contains $T_a$ as an induced minor by just mapping the non-root vertices of the complete binary trees added to $T_0$ to the corresponding vertices of complete binary trees added to $G_0$.

It remains to show that if $G_a$ contains $T_a$ as an induced minor, then $G_0$ contains $T_0$ as an anchored induced minor.
For this, we will argue that the natural places for the added binary trees are in fact the only possible places for the binary trees in any induced minor model of $T_a$ in $G_a$.
For this, we start with the following auxiliary claim.

\begin{claim}
\label{lem:mainhardness:claimpw}
The graph $G_i$ does not contain $B_{h+2i+1}$ as an induced minor.
\end{claim}
\begin{claimproof}
We prove this by induction on $i$.
For $G_0$ this holds by definition.
Let $i \ge 1$ and for the sake of contradiction, assume that $G_i$ contains $B_{h+2i+1}$ as an induced minor.
Now, by removing from $B_{h+2i+1}$ the vertex whose branch set contains $v_i$, we have that either $G_i[L_i]$ or $G_i[R_i]$ contains $B_{h+2i}$ as an induced minor.
However, $G_i[L_i]$ does not contain it by induction (it is a subgraph of $G_{i-1}$), and $G_i[R_i]$ does not contain it because it has less vertices than $B_{h+2i}$.
\end{claimproof}

Next we show that the branch set of the root of the complete binary tree contains $v_i$.

\begin{claim}
\label{lem:mainhardness:claimroot}
In any induced minor model of $B_{h+2i}$ in $G_i$, the branch set of the root of $B_{h+2i}$ contains $v_i$.
\end{claim}
\begin{claimproof}
First, we have that some branch set of an induced minor model of $B_{h+2i}$ in $G_i$ must intersect $v_i$, because neither $G_i[L_i]$ nor $G_i[R_i]$ contains $B_{h+2i}$ as an induced minor ($G_i[L_i]$ by \Cref{lem:mainhardness:claimpw} and $G_i[R_i]$ because it has less vertices than $B_{h+2i}$).
If this branch set is the branch set of the root we are done, so suppose this is the branch set of some other vertex of $B_{h+2i}$.
Without loss of generality, let us assume that this vertex of $B_{h+2i}$ is in the right subtree of $B_{h+2i}$.
Now, the induced minor model for the left subtree must be contained in either $L_i$ or $R_i$.
As it is an induced minor model of $B_{h+2i-1}$, by \Cref{lem:mainhardness:claimpw} it cannot be contained in $L_i$, so it is contained in $R_i$.
As $G_i[R_i]$ has two connected components and both are isomorphic to $B_{h+2i-1}$, we have that the left subtree must be mapped into either of these components with all branch sets of size one.
However, in that case the branch set of the root must contain $v_i$.
\end{claimproof}

Then, we show that the natural place for the complete binary tree is the only possible.

\begin{claim}
\label{lem:mainhardness:claimall}
Any induced minor model of $B_{h+2i}$ in $G_i$ must contain all vertices of $\{v_i\} \cup R_i$, and all branch sets except the root must be of size $1$ and be contained in $R_i$. 
\end{claim}
\begin{claimproof}
By \Cref{lem:mainhardness:claimroot}, the branch set of the root must contain $v_i$.
Now, by connectivity, both the model of the left subtree of $B_{h+2i}$ and the right subtree of $B_{h+2i}$ must be completely contained in either $L_i$ or $R_i$, but by \Cref{lem:mainhardness:claimpw} neither of them can be contained in $L_i$, so both of them are contained in $R_i$.
Now, as $|R_i| = |V(B_{h+2i})|-1$, we have that all of the vertices in $R_i$ must be used for the left and right subtrees of $B_{h+2i}$.
\end{claimproof}

Now, \Cref{lem:mainhardness:claimall} implies that any induced minor model of $T_a$ in $G_a$ must map the induced subgraph $T_a[R^T_a]$ to the induced subgraph $G_a[R_a]$, with branch sets of size $1$.
By iterating this claim for $i<a$, we get that for all $i \in [a]$ any induced minor model of $T_a$ in $G_a$ must map the induced subgraph $T_a[R^T_i]$ to the induced subgraph $G_a[R_i]$ with branch sets of size $1$.
From this, it follows that any induced minor model of $T_a$, after removing the branch sets of vertices $V(T_a) \setminus V(T_0)$, must result in an anchored induced minor model of $T_0$ in $G_0$.
\end{proof}

We note that all of our reductions are polynomial-time reductions, so our hardness proof implies also NP-hardness of induced minor testing for a fixed tree $T$.
We also note again that under randomized reductions our lower bound could be improved to $2^{o(n/\log^2 n)}$ by improving \Cref{lem:bincsphardness} by using \Cref{lem:fromcftoimm}.

\section{Conclusion}
\label{sec:concl}
We gave a $\OO_H(\sqrt{m})$ separator theorem for $H$-induced-minor-free graphs and applied it for obtaining $2^{\OO_H(n^{2/3} \log n)}$ time algorithms on such graphs.
We used it also to obtain $2^{\OO_{H}(n^{2/3} \log n)}$ time algorithm for $H$-induced minor testing for graphs $H$ whose minimal induced minor models are guaranteed to have degeneracy bounded as a function of $H$.
We then showed that this bounded degeneracy condition is essential for obtaining such algorithms by giving a $2^{o(n / \log^3 n)}$ lower bound for $T$-induced minor testing for some fixed tree $T$.

Let us make some further remarks about (potential) applications of our separator theorem.
The result of K{\"{u}}hn and Osthus~\cite{DBLP:journals/combinatorica/KuhnO04a} implies that there is a function $f(H,t)$ so that $H$-induced-minor-free graphs that do not contain $K_{t,t}$ as a subgraph have at most $f(H,t) \cdot n$ edges.
This combined with \Cref{thm:septheorem} immediately gives a $\OO_{H,t}(\sqrt{n})$ separator theorem for $K_{t,t}$-subgraph-free $H$-induced-minor-free graphs, implying that such graphs have treewidth at most $\OO_{H,t}(\sqrt{n})$.
Interestingly, a construction of Davies~\cite{davies-example} shows that such graphs can attain treewidth $\Omega(\sqrt{n})$ even when $t = 2$ and $H$ is a $5 \times 5$-grid.

We focused on the applications of our separator theorem to subexponential time algorithms.
For string graphs, a similar separator theorem has been applied also for other types of structural~\cite{DBLP:journals/cpc/FoxP10,DBLP:journals/cpc/FoxP14} and algorithmic~\cite{DBLP:conf/soda/FoxP11} results, so it would be interesting to explore the applications of our separator theorem also in other settings.

Motivated by the result of~\cite{DBLP:journals/jctb/Korhonen23} that bounded-degree graphs excluding a $k \times k$-grid as an induced minor also exclude a $f(k) \times f(k)$-grid as a minor for some function $f$, one could ask if a similar result would hold for excluding subdivided cliques.
Here, the answer is however negative, because by adding a crossings inside every cell of a $n \times n$-grid (see \Cref{fig:crossgrid}) we obtain a graph with maximum degree $8$ that contains large (of size polynomial in $n$) subdivided cliques as minors, but excludes a subdivided $K_5$ as an induced minor.

\begin{figure}[htb]
\begin{center}
\includegraphics[width=0.2\textwidth]{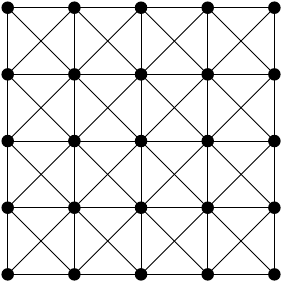}
\caption{A $5 \times 5$-grid with crossings added inside every cell.}
\label{fig:crossgrid}
\end{center}
\end{figure}

There is a trivial algorithm for induced minor testing that works in time $2^{\OO_{H}(n)}$, and while our hardness proof excludes $2^{\OO_{H}(n/\log^3 n)}$ time algorithms, we believe that there is no $2^{o(n)}$ time algorithm for $T$-induced minor testing for some fixed tree $T$.
In particular, we conjecture that there is no $2^{o(n)}$ time algorithm for the problem of \Cref{lem:chaininduced}, which by our reductions would imply that there exists a fixed tree $T$ so that there is no $2^{o(n)}$ time algorithm for testing for $T$ as an induced minor.
It seems reasonable that the techniques of \Cref{subsec:binshift} could be improved to improve our lower bound to $2^{o(n / \log n)}$, but for improving it to $2^{o(n)}$ other techniques could be needed.

While finding induced minors whose models have unbounded degeneracy appears to be hard, we are not aware of even NP-hardness results for finding induced minors $H$ whose every edge is adjacent to a vertex of degree at most $2$.
In particular, we ask if the problem \textsc{Induced $k$-Disjoint Paths} is NP-hard on $H$-induced-minor-free graphs for some fixed $k$ and $H$.
We note that \textsc{Induced $2$-Disjoint Paths} is known to be NP-hard on general graphs~\cite{DBLP:journals/dm/Bienstock91}.

Another open problem is whether the running times of the algorithms of \Cref{cor:algos} could be significantly improved, i.e., to $2^{\OO_H(n^{2/3-\varepsilon})}$ for some $\varepsilon > 0$.
Obtaining a positive answer for this appears to be difficult, since no such algorithms are known even for string graphs.
For a lower bound under ETH, one should construct $H$-induced-minor-free graphs for which the techniques of Marx and Pilipczuk~\cite{DBLP:conf/esa/MarxP15} are not applicable to, which also appears challenging.  

\bibliographystyle{alpha}
\bibliography{book_kernels_fvf}

\end{document}